\documentclass[aap,MSNbibl,seceqn,dvips]{arximspdf}
\usepackage{graphicx}
\doi{10.1214/12-AAP902} %kopijuoti is PTS
\volume{24}
\issue{1}
\pubyear{2014}
\firstpage{1}
\lastpage{33}

\makeatletter
\newtheorem{theorem}{Theorem}[section]
\newtheorem{lemma}[theorem]{Lemma}
\newtheorem{corollary}[theorem]{Corollary}

\newproclaim{remark}[theorem]{Remark}
\newproclaim{example}[theorem]{Example}
\newproclaim{definition}[theorem]{Definition}
\newproclaim{notation}[theorem]{Notation}

\makeatother

\begin{document}
\begin{frontmatter}

\title{Killed Brownian motion with a prescribed lifetime distribution
and models of default}
\runtitle{Killed Brownian motion}

\begin{aug}
\author[A]{\fnms{Boris} \snm{Ettinger}\ead[label=e1]{ettinger@princeton.edu}},
\author[B]{\fnms{Steven N.} \snm{Evans}\corref{}\ead[label=e2]{evans@stat.berkeley.edu}\thanksref{t1}}
\and
\author[C]{\fnms{Alexandru} \snm{Hening}\ead[label=e3]{hening@stats.ox.ac.uk}}
\runauthor{B. Ettinger, S. N. Evans and A. Hening}
\affiliation{Princeton University, University of California and University of Oxford}
\address[A]{B. Ettinger\\
Department of Mathematics\\
Princeton University\\
Fine Hall, Washington Road\\
Princeton, New Jersey 08544-1000\\
USA\\
\printead{e1}}
\address[B]{S. N. Evans\\
Department of Statistics\\
University of California\\
367 Evans Hall \#3860\\
Berkeley, California 94720-3860\\
USA\\
\printead{e2}} %adresu isvedimo komanda gale!
\address[C]{A. Hening\\
Department of Statistics\\
University of Oxford\\
1 South Parks Road\\
Oxford OX1 3TG\\
United Kingdom\\
\printead{e3}}
\end{aug}
\thankstext{t1}{Supported in part by NSF grant DMS-09-07630.}

% HISTORY:
\received{\smonth{12} \syear{2011}}
\revised{\smonth{9} \syear{2012}}

% ABSTRACT
%
\begin{abstract}
The inverse first passage time problem asks whether, for a Brownian
motion $B$ and a nonnegative random variable $\zeta$, there exists a
time-varying barrier $b$ such that $\mathbb{P}\{B_s > b(s), 0 \leq s
\leq t\} = \mathbb{P}\{\zeta> t\}$. We study a ``smoothed'' version of
this problem and ask whether there is a ``barrier'' $b$ such that
$\mathbb{E}[\exp(-\lambda\int_0^t \psi(B_s - b(s))\,ds)] =
\mathbb{P}\{\zeta> t\}$, where $\lambda$ is a killing rate parameter,
and $\psi\dvtx\mathbb{R} \to[0,1]$ is a nonincreasing function. We
prove that if $\psi$ is suitably smooth, the function $t \mapsto
\mathbb{P}\{\zeta> t\}$ is twice continuously differentiable, and the
condition $0 < -\frac{d \log\mathbb{P}\{\zeta> t\}}{dt} < \lambda$
holds for the hazard rate of $\zeta$, then there exists a unique
continuously differentiable function $b$ solving the smoothed problem.
We show how this result leads to flexible models of default for which
it is possible to compute expected values of contingent claims.
\end{abstract}

% KEYWORDS

\pdfkeywords{60J70, 91G40, 91G80, Credit risk, inverse first passage time problem, killed Brownian motion,
Cox process, stochastic intensity, Feynman--Kac formula}

\end{frontmatter}

%s1 ###
%s1 #&#
\section{Introduction}\label{sec1}
Investors are exposed to credit risk, or counterparty risk, due to the
possibility that one or more counterparties in a financial agreement
will default, that is, not honor their obligations to make certain
payments. Counterparty risk has to be taken into account when pricing a
transaction or portfolio, and it is necessary to model the occurrence
of default jointly with the behavior of asset values.

The default time is sometimes modeled as the first passage time of a
credit index process below a barrier. Black and Cox \cite{BC76} were
among the first to use this approach. They define the time of default
as the first time the ratio of the value of a firm and the value of its
debt falls below a constant level, and they model debt as a zero-coupon
bond and the value of the firm as a geometric Brownian motion. In this
case, the default time has the distribution of the first-passage time
of a Brownian motion (with constant drift) below a certain barrier.

Hull and White \cite{HW01} model the default time as the first time a
Brownian motion hits a given time-dependent barrier. They show that
this model gives the correct market credit default swap and bond prices
if the time-dependent barrier is chosen so that the first passage time
of the Brownian motion has a certain distribution derived from those
prices. Given a distribution for the default time, it is usually
impossible to find a closed-form expression for the corresponding
time-dependent barrier, and numerical methods have to be used.

We adopt a perspective similar to that of Hull and White \cite{HW01}.
Namely, we model the default time as
%
%
%e1.1 ###
%
%e1.1 #&#
\begin{equation}
\label{erandomtime} \tau:=\inf \biggl\{t>0\dvtx\lambda\int_0^t
\psi \bigl(Y_s-b(s) \bigr)\,ds> U \biggr\},
\end{equation}
where the diffusion $Y$ is some credit index process, $U$ is an
independent mean one exponentially distributed random variable, $0\leq
\psi\leq1$ is a suitably smooth, nonincreasing function with
$\lim_{x \to-\infty} \psi(x) = 1$ and $\lim_{x \to+ \infty} \psi
(x) = 0$, and \mbox{$\lambda>0$} is a rate parameter. Then
%
%
%e1.2 ###
%
%e1.2 #&#
\begin{equation}
\label{ePtau} \mathbb{P}\{\tau>t\}=\mathbb{E} \biggl[\exp \biggl(-\lambda\int
_0^t\psi \bigl(Y_s-b(s) \bigr)\,ds
\biggr) \biggr].
\end{equation}
The random time $\tau$ is a ``smoothed-out'' version of the stopping
time of Hull and White; instead of killing $Y$ as soon at it crosses
some sharp, time-dependent boundary, we kill $Y$ at rate $\lambda\psi(y
- b(t))$ if it is in state $y \in\mathbb{R}$ at time $t \geq 0$. That
is,
\[
\lim_{\Delta t \downarrow0} \mathbb{P} \bigl\{\tau\in(t,t+\Delta t) \mid
(Y_s)_{0 \leq s \leq t}, \tau> t \bigr\} / \Delta t = \lambda\psi
\bigl(Y_t - b(t) \bigr).
\]
When the credit index value $Y_t$ is large, corresponding to a time $t$
when the counterparty is in sound financial health, the killing rate
$\lambda\psi(Y_t - b(t))$ is close to $0$ and default in an ensuing
short period of time is unlikely, whereas the killing rate is close to
its maximum possible value, $\lambda$, when $Y_t$ is low and default is
more probable. Note that if we consider a family of $[0,1]$-valued,
nonincreasing functions $\psi$ that converges to the indicator
function of the set $\{x \in\mathbb{R} \dvtx x < 0\}$ and $\lambda$
tends to $\infty$, then the corresponding stopping time $\tau$
converges to the Hull and White stopping time $\inf\{t > 0 \dvtx Y_t <
b(t)\}$.

The hazard rate of the random time $\tau$ is
%
%
%e1.3 ###
%
%e1.3 #&#
\begin{eqnarray}
\label{ehazard} && \frac{\mathbb{P}\{\tau\in dt \mid\tau> t\}}{dt}
\nonumber
\\
&&\qquad:= \lim_{\Delta t\downarrow0} \frac{\mathbb{P}\{\tau\in
(t,t+\Delta t)\}}{\Delta t\mathbb{P}\{\tau> t\}}
\nonumber
\\
&&\qquad= \lim_{\Delta t\downarrow0} \frac{\mathbb{P} \{\lambda
\int_0^t \psi(Y_s-b(s))\,ds\leq U\leq\lambda\int_0^{t+\Delta t}\psi
(Y_s-b(s))\,ds \}}{\Delta t\mathbb{P} \{\lambda\int_0^t
\psi(Y_s-b(s))\,ds\leq U \}}
\\
&&\qquad= \lim_{\Delta t\downarrow0} \frac{\mathbb{E} [e^{-\lambda
\int_0^t \psi(Y_s-b(s))\,ds}-e^{- \lambda\int_0^{t+\Delta t}\psi
(Y_s-b(s))\,ds } ]}{\Delta t\mathbb{E} [\exp(-\lambda
\int_0^t \psi(Y_s-b(s))\,ds ) ]}
\nonumber
\\
&&\qquad= \frac{ \lambda\mathbb{E} [\psi(Y_t-b(t)) \exp
(-\lambda\int_0^t \psi(Y_s-b(s))\,ds ) ]}{\mathbb
{E} [\exp(-\lambda\int_0^t \psi(Y_s-b(s))\,ds
) ]}.
\nonumber
\end{eqnarray}
%
%A natural choice for $\psi$ would be $\psi=\mathbf{1}_{(-\infty,0]}$,
%where
% \begin{aligned}
% 0,~~~x>0 \\
% 1,~~~x\leq0 \\
% \end{aligned}
% \right.
%Then the hazard would become

On the other hand, suppose that $\zeta$ is a
nonnegative random variable with
survival function $t \mapsto G(t):= \mathbb{P}\{\zeta> t\}$.
Writing $g$ for the derivative of $G$, the corresponding hazard rate is
\[
-\frac{g(t)}{G(t)} = - \frac{d}{dt} \log G(t).
\]
As a result, a necessary condition for a function $b$ to exist such that
the corresponding random time $\tau$ has the same distribution as
$\zeta$
is that
%
%
%e1.4 ###
%
%e1.4 #&#
\begin{equation}
\label{Ehazardbound} 0 < -g(t) < \lambda G(t),\qquad t\geq0.
\end{equation}

We show in Theorem~\ref{tglobalEandU} that if $Y$ is a Brownian motion
with a given suitable random initial condition, assumption
(\ref{Ehazardbound}) holds, and the survival function $G$ is twice
continuously differentiable, then there is a unique differentiable
function~$b$ such that the stopping time $\tau$ has the same
distribution as $\zeta$. In particular, we establish that the function
$b$ can be determined by solving a system consisting of a parabolic
linear PDE with coefficients depending on $b$ and a nonlinear ODE
for~$b$ with coefficients depending on the solution of the PDE. Note from
(\ref{ePtau}) that changing the function $b$ on a set with Lebesgue
measure zero does not affect the distribution of $\tau$, and so we have
to be careful when we talk about the uniqueness of $b$. This minor
annoyance does not appear if we restrict to continuous $b$.

%Can we produce any hazard that is bounded by
%$\lambda$ by choosing the function $b$ appropriately (for a given $t$,
%the right
%hand side varies between 0 and $\lambda$ as $b(t)$ varies between $-
%$+\infty$)? Moreover, does a given hazard lead to a unique choice of
%$b$? Note that if we require $b$ only to be measurable then equation
%measure zero set.

%In the next sections we first investigate the existence and uniqueness
%of a barrier $b$ for a given survival function $G$ when the function $

In Theorem~\ref{Thardbarrier} we give an analogue of the existence part
of the above result when $\psi$ is the indicator of the set $\{x \in
\mathbb{R} \dvtx x < 0\}$.

Having proven the existence and uniqueness of a barrier $b$, we
consider the pricing of certain contingent claims in
Section~\ref{Spricing}. For simplicity, we take the asset price
$(X_t)_{t \geq 0}$ to be a geometric Brownian motion
\[
\frac{dX_t}{X_t}=\mu\,dt+\sigma\,dW_t,
\]
where $W$ is a standard Brownian motion. We take the
credit index $(Y_t)_{t\geq0}$
to be given by
\[
dY_t=dB_t,
\]
where $B$ is another standard Brownian motion, and take the default
time to be given by (\ref{erandomtime}), where the exponential random
variable $U$ is independent of the asset price $X$ and the credit index
$Y$. We assume that the Brownian motions $W$ and $B$ are correlated;
that is, that their covariation is $[B,W]_t=\rho t$ for some constant
$\rho\in[-1,1]$. We consider claims with a payoff of the form
$F(X_T)1\{\tau>T\}$ for some fixed maturity $T$. We show how it is
possible to compute conditional expected values such as
\[
\mathbb{E} \bigl[F(X_T) 1\{\tau>T\} \mid(X_s)_{0\leq s\leq t},
\tau>t \bigr].
\]

In Section~\ref{Snumerical} we report the results of some experiments
where we solved the PDE/ODE system for the barrier $b$ numerically.
Finally, in Section~\ref{Scalibration}, we follow~\cite{DP11} to
demonstrate how it is possible to use market data on credit default
swap prices to determine the survival function $G$.

%s1.1 ###
%s1.1 #&#
\subsection{The FPT and IFPT problems}

We end the this \hyperref[sec1]{Introduction} with a brief discussion
of the literature dealing with first passage times of diffusions across
time-dependent barriers.

Consider a Brownian motion $(B_t)_{t\geq0}$ defined on a filtered
probability space $(\Omega,\mathbb{P}, \mathcal{F},(\mathcal
{F}_t)_{t\geq0})$ which satisfies the usual conditions. Define the
diffusion $(Y_t)_{t \geq 0}$ via the SDE
\[
dY_t = \mu(Y_t,t)\,dt + \sigma(Y_t,t)
\,dB_t,
\]
where we assume that the coefficients
$\mu\dvtx\mathbb{R}\times\mathbb{R}_+\rightarrow\mathbb{R}$ and
$\sigma\dvtx\mathbb{R}\times\mathbb{R}_+\rightarrow\mathbb{R}_+$ are
such that the SDE has a unique strong solution.

For a Borel function
$b\dvtx\mathbb{R}_+\rightarrow\overline{\mathbb{R}}:=
\mathbb{R}\cup\{\pm\infty\}$, the first passage time of the diffusion
process $Y$ below the barrier $b$ is the stopping time
%
%
%e1.5 ###
%
%e1.5 #&#
\begin{equation}
\label{eboundarycrossingtime} \tilde{\tau}=\inf \bigl\{t>0\dvtx Y_t
< b(t) \bigr\}.
\end{equation}
The following two problems related to this notion have been discussed
in the literature.

\textit{The first passage time problem} (\textit{FPT}): For a given barrier
$b\dvtx\mathbb{R}_+ \rightarrow\overline{\mathbb{R}}$, compute the
survival function~$G$ of the first time that $X$ goes below $b$; that
is, find
%
%
%e1.6 ###
%
%e1.6 #&#
\begin{equation}
\label{esurvival} G(t):=\mathbb{P}\{\tilde{\tau}>t\}, \qquad t \geq 0.
\end{equation}

\textit{The inverse first passage time problem} (\textit{IFPT}): For a given
survival function~$G$, does there exist a barrier $b$ such that
$G(t)=\mathbb{P}\{\tilde{\tau}>t\}$ for all $t\geq0$?

A large class of first passage time problems may be
solved within a PDE framework.
Let
$u(x,t)=\frac{\partial}{\partial x} \mathbb{P}\{Y_t\leq x, \tilde
{\tau}>t\}$
be the sub-probability density of the diffusion $Y$ killed at $\tilde
\tau$.
Then, by the Kolmogorov forward equation, $u$ satisfies
%
%
%e1.7 ###
%
%e1.7 #&#
\begin{equation}
\label{eKolmFPT} \cases{u_t(x,t) = \frac{1}{2} \bigl(
\sigma^2 u \bigr)_{xx} - (\mu u)_x, &\quad
$x>b(t)$, $t>0$, \vspace*{2pt}
\cr
u(x,t) =0, &\quad$x\leq b(t)$, $t>0$,
\vspace*{2pt}
\cr
u(x,0)= f(x), &\quad$x \in\mathbb{R}$,}
\end{equation}
where $f$ is the probability density of $Y_0$.
For nice enough functions $b$ this system has a unique solution, and we
can express the survival probability as
\[
G(t) =\mathbb{P}\{\tilde{\tau}>t\} = \int_{b(t)}^{\infty}
u(x,t)\,dx, \qquad t\geq0.
\]
This approach is used in \cite{L86,V09}
%Lerche \cite{L86} and Valov \cite{V09}
to get closed form solutions for some classes of boundaries.
An integral equation technique is used in
%Peskir \cite{P2002a}, Peskir and Shiryaev \cite{PS06} and Valov (
\cite{P2002a,P2002b,PS06,V09} to find the derivative $g(t) = G'(t)$ in
the FPT problem for a Brownian motion. Writing $\Psi(z):=
\int_z^\infty\frac{1}{\sqrt{2 \pi}} \exp(-\frac{x^2}{2} )\,dx$, the
derivative $g$ satisfies a Volterra integral equation of the first kind
of the form
\[
\Psi \biggl(\frac{b(t)}{\sqrt{t}} \biggr) = -\int_0^t
\Psi \biggl( \frac{b(t)-b(s)}{\sqrt{t-s}} \biggr)g(s)\,ds.
\]
This and other such integral equations can be used
to find $g$ numerically.

Shiryaev is generally credited with introducing the IFPT problem in
1976 (we have not been able to find an explicit reference).
%This problem is important in credit risk management as this setting is
%useful for modeling the default time of a company. If the diffusion
%$Y_t$ is an indicator of the firm's financial health, and $G$ is the
%survival function,
%which we can estimate from the prices of various financial
%instruments, the questions is whether we can find a default barrier
%that yields the given survival distribution.
Most authors have investigated numerical methods for finding the
boundary. Details can be found in \cite{HW00,HW01,IK02,ZSP03}.
%in Hull and White \cite{HW00}, Hull and White \cite{HW01}, Iscoe and
%Kreinin \cite{IK02}, Zucca et al \cite{ZSP03}.
%Avellaneda and Zhu \cite{AZ01}
It is shown in \cite{AZ01} that for sufficiently smooth boundaries the
density $u(x,t)$ and the boundary $b(t)$ are a solution of the
following free boundary problem:
%
%
%e1.8 ###
%
%e1.8 #&#
\begin{equation}
\label{eKolmIFPT} \cases{\displaystyle u_t(x,t) = \frac{1}{2}
\bigl(\sigma^2u \bigr)_{xx} - (\mu u)_x, &\quad
$x>b(t)$, $t>0$, \vspace*{2pt}
\cr
u(x,t) =0, &\quad$x\leq b(t)$, $t>0$,
\vspace*{2pt}
\cr
u(x,0)= f(x), &\quad$x \in\mathbb{R}$, \vspace*{2pt}
\cr
\displaystyle G(t) = \int_{b(t)}^{\infty} u(x,t)\,dx, &
\quad $t\geq0$,}
\end{equation}
where $f$ is again the probability density of $Y_0$. The existence and
uniqueness of a viscosity solution of (\ref{eKolmIFPT}) is established
in \cite{CCCS11} along with upper and lower bounds on the asymptotic
behavior of $b$. That paper also shows that this $b$ does in fact
produce a boundary that gives the survival function $G$. To our
knowledge it has not be proven that a strong solution to the system
(\ref{eKolmIFPT}) exists, nor that there is a smooth $b$ solving the
IFPT.

A variation of the IFPT is studied in \cite{DP11,DaPi13}. There the barrier is
fixed at zero (i.e., $b\equiv0$), and it is the volatility parameter
$\sigma(\cdot, \cdot)$, that is, allowed to vary. The authors show that
this problem admits an explicit solution for every differentiable
survival function.

%s2 ###
%s2 #&#
\section{Global existence and uniqueness}

Suppose for the remainder of this paper that $Y_t:= Y_0+B_t$, where
$(B_t)_{t \geq 0}$ is a standard Brownian motion, and $Y_0$ is a random
variable, independent of $B$ and with density $f \in C^2(\mathbb{R})$.
In this case, (\ref{ePtau})~is
\[
G(t) = \int_{\mathbb{R}} \mathbb{E} \biggl[\exp \biggl(-\lambda \int
_0^t \psi \bigl(x + B_s - b(s)
\bigr)\,ds \biggr) \biggr]f(x)\,dx,
\]
which, by time reversal, becomes
\[
G(t) = \int_{\mathbb{R}} \mathbb{E} \biggl[\exp \biggl(-\lambda \int
_0^t \psi \bigl(x + B_{t-s} - b(s)
\bigr)\,ds \biggr) f(x + B_t) \biggr]\,dx.
\]
Set
%
%
%e2.1 ###
%
%e2.1 #&#
\begin{equation}
u(x,t):= \mathbb{E} \biggl[\exp \biggl(-\lambda\int_0^t
\psi \bigl(x + B_{t-s} - b(s) \bigr)\,ds \biggr) f(x + B_t)
\biggr].
\end{equation}
That is, $u$ is the sub-probability density of
$Y$ killed at the random time $\tau$.
It is well known that if $u$ is smooth enough, then $u$ is the unique
solution of the PDE
\[
\cases{\displaystyle u_t(x,t) = \tfrac{1}{2}u_{xx}(x,t)
- \lambda\psi \bigl(x-b(t) \bigr) u(x,t), &\quad $x \in\mathbb{R}$, $t>0$,
\vspace*{2pt}
\cr
u(x,0) =f(x), &\quad $x \in\mathbb{R}$.}
\]
Any solution to this PDE satisfies
%
%
%e2.2 ###
%
%e2.2 #&#
\begin{equation}
\label{evanishing} \lim_{x\rightarrow\pm\infty}u(x,t) = \lim
_{x\rightarrow\pm\infty
}u_x(x,t) = 0, \qquad t > 0.
\end{equation}
Our question as to whether we can find a ``barrier'' $b$ giving us the
survival function~$G$~is now equivalent to whether the system
%
%
%e2.3 ###
%
%e2.3 #&#
\begin{equation}
\label{ePDE} \cases{\displaystyle u_t(x,t) = \frac{1}{2}u_{xx}(x,t)
- \lambda\psi \bigl(x-b(t) \bigr) u(x,t), &\quad $x\in\mathbb{R}$, $t> 0$,
\vspace*{2pt}
\cr
u(x,0) =f(x), &\quad $x\in\mathbb{R}$, \vspace*{2pt}
\cr
\displaystyle\int_{\mathbb{R}}u(x,t)\,dx = G(t), &\quad $t\geq0$,}
\end{equation}
has solutions $(u,b)$. Differentiating the third equation from
(\ref{ePDE}) with respect to $t$ and then using the first equation
together with an integration by parts, we get that
%
%
%e2.4 ###
%
%e2.4 #&#
\begin{equation}
\label{egpsi} -g(t) = \lambda\int_{\mathbb{R}}\psi \bigl(x-b(t)
\bigr)u(x,t)\,dx,
\end{equation}
where we recall that $g(t) = G'(t)$.
A second differentiation in $t$ followed by another integration by
parts yields
\begin{eqnarray}
\label{eg} \hspace*{25pt} && g'(t) - \lambda^2\int
_{\mathbb{R}}\psi^2 \bigl(x-b(t) \bigr) u(x,t)\,dx
\nonumber
\\
&&\qquad = \lambda\int_{\mathbb{R}} \psi_x \bigl(x-b(t)
\bigr) u(x,t)b'(t)\,dx
\nonumber
- \lambda/2 \int
_{\mathbb{R}} \psi \bigl(x-b(t) \bigr)u_{xx}(x,t)\,dx
\nonumber
\\[-8pt]
\\[-8pt]
&&\qquad = \lambda\int_{\mathbb{R}} \psi_x \bigl(x-b(t)
\bigr)u(x,t) b'(t)\,dx + \lambda/2 \int_{\mathbb{R}}
\psi_x \bigl(x-b(t) \bigr) u_x(x,t)\,dx
\nonumber
\\
&&\qquad = \lambda\int_{\mathbb{R}} \psi_x \bigl(x-b(t)
\bigr)u(x,t) b'(t)\,dx
\nonumber
- \lambda/2 \int
_{\mathbb{R}} \psi_{xx} \bigl(x-b(t) \bigr) u(x,t)\,dx.
\nonumber
\end{eqnarray}
Note that (\ref{eg}) may be rearranged to give an ODE for $b$ of the
form $b'(t) = \Theta(b(t),t)$, where the function $\Theta$ is
constructed from the function $u$ (which, of course, depends in turn
on~$b$). Re-writing this integral equation in the form $b(t) = b(0) +
\int_0^t \Theta(b(s),s)\,ds$ leads to the following theorem, our main
result.

%
%
%th2.1 #&#
\begin{theorem}\label{tglobalEandU}
Suppose the following:
\begin{itemize}
\item The survival function $G$ is twice continuously
differentiable with first and second derivatives $g$ and $g'$
and $0 < -g(t) < \lambda G(t)$ for all $t \geq 0$ for some
constant $\lambda
> 0$.

\item The initial density $f$ satisfies
$\int_{\mathbb{R}}f(x)\,dx=1$, $f(x)>0$ for all
$x\in\mathbb{R}$, $f\in C^2(\mathbb{R})$, and the functions
$f$, $f'$, $f''$ are bounded.

\item The function $\psi$ is nonincreasing and belongs to
$C^3(\mathbb{R})$, and for some $h>0$, $\psi(x)=1$ for
$x\leq-h$ and $\psi(x)=0$ for $x\geq h$.
\end{itemize}
Then, there exists a unique continuously differentiable function
$b\dvtx[0,\infty)\rightarrow\mathbb{R}$ such that the following three
equations hold:
%
%
%e2.5 ###
%
%e2.5 #&#
%e2.6 #&#
\begin{eqnarray}
\label{e1} G(t) &=& \int_{\mathbb{R}} \mathbb{E} \biggl[\exp
\biggl(- \lambda\int_0^t \psi \bigl(x +
B_u - b(u) \bigr)\,du \biggr) \biggr]f(x)\,dx,
\\
-g(t)&=& \lambda\int_{\mathbb{R}} \mathbb{E} \biggl[\exp
\biggl(-\lambda\int_0^t \psi \bigl(x
+B_u - b(u) \bigr)\,du \biggr)
\nonumber
\\[-8pt]
\label{e2}
\\[-8pt]
&&\hspace*{97pt} {}\times \psi \bigl(x + B_t -b(t) \bigr)
\biggr]f(x)\,dx
\nonumber
\end{eqnarray}
and
%
%
%e2.6 ###
%
%e2.7 #&#
\begin{eqnarray}\quad
\label{e3} b(t) &=& b(0)
\nonumber
\\
&&{\fontsize{10.4pt} {12pt}\selectfont{\mbox{$\displaystyle{} + \int
_0^t \biggl(\frac{g'(s)
-\lambda^2\int_{\mathbb{R}} \mathbb{E} [\psi^2(x+B_s-b(s))
e^{-\lambda\int_0^s \psi(x+B_r-b(r))\,dr} ] f(x)\,dx
}{\lambda\int_{\mathbb{R}} \mathbb{E} [\psi_x(x+B_s-b(s))
e^{-\lambda\int_0^s \psi(x+B_r-b(r))\,dr} ] f(x)\,dx}$}}}
\\
&&\hspace*{41pt} {\fontsize{10.4pt} {12pt}\selectfont{\mbox{$\displaystyle{} +
\frac{\lambda/2 \int_{\mathbb{R}} \mathbb{E} [\psi_{xx}(x+B_s-b(s))
e^{-\lambda\int_0^s \psi(x+B_r-b(r))\,dr} ] f(x)\,dx
} {\lambda\int_{\mathbb{R}} \mathbb{E} [\psi_x(x+B_s-b(s))
e^{-\lambda\int_0^s \psi(x+B_r-b(r))\,dr} ] f(x)\,dx
} \biggr)\,ds$}}}\hspace*{-9pt}
\nonumber
\end{eqnarray}
for all $t \geq 0$.
\end{theorem}

\begin{pf}
From now on we assume for ease of notation that $\lambda=1$. The
modifications necessary for general $\lambda$ are straightforward. The
proof will be via a sequence of lemmas, all of them assuming the
hypotheses of Theorem~\ref{tglobalEandU} (with $\lambda= 1$). We start
with the following simple observation.

%
%
%le2.2 #&#
\begin{lemma}\label{lIVP}
Suppose that
\[
G(t)=\int_{\mathbb{R}}u(x,t)\,dx
\]
for some continuous function $u\dvtx\mathbb{R}\times\mathbb{R}_+ \to
\mathbb{R}$ such that $u(x,t)>0$ for $x\in\mathbb{R}$, $t\geq0$.
%Assume we also have
%and
%0<-g(t)<G(t)
%for all $t\geq0$.
Then, for each $t \geq 0$ there exists a unique $b(t) \in\mathbb{R}$
such that
\[
-g(t) = \int_{\mathbb{R}} \psi \bigl(x - b(t) \bigr) u(x,t)\,dx.
\]
\end{lemma}

\begin{pf}
Set
\[
F(t,z) = \int_{\mathbb{R}} \psi(x - z) u(x,t)\,dx.
\]
Then
\begin{eqnarray*}
\lim_{z \to-\infty} F(t,z) &=& \int_{\mathbb{R}} u(x,t)\,dx
=G(t),
\\
\lim_{z \to+\infty} F(t,z) &=& 0
\end{eqnarray*}
and, by assumption,
\[
0<-g(t)<G(t).
\]
Furthermore, $F$ is continuous and strictly decreasing in $z$. So, by
the intermediate value property, we can find a unique $b(t) \in\mathbb
{R}$ such that $F(t,b(t))=-g(t)$.
\end{pf}

%
%
%le2.3 #&#
\begin{lemma}[(Global uniqueness)]\label{luniqueness}
Suppose there exist continuous functions $b_1, b_2$ such that equations
(\ref{e1}), (\ref{e2}) and (\ref{e3}) are satisfied for $b=b_1$ and
$b=b_2$. Then $b_1(t)=b_2(t)$ for all $t\geq0$.
\end{lemma}

\begin{pf}
Recall that we are assuming $\lambda= 1$ to simplify notation.

Suppose that $b_1$ and $b_2$ are two continuous solutions of
(\ref{e1}), (\ref{e2}) and (\ref{e3}). It follows from Lemma \ref{lIVP}
and (\ref{e2}) that $b_1(0)=b_2(0)$. Set $V:= \inf\{t \geq 0 \dvtx
b_1(t) \ne b_2(t)\}$, and suppose that $V<\infty$.

Define $\tilde f \in C^2(\mathbb{R})$ by
\[
\tilde f(y)\,dy:= \int_{\mathbb{R}} \mathbb{E} \bigl[\mathbf{1} \{x
+ B_V \in dy\} e^{-\int_0^V \psi(x+B_r-b(r))\,dr} \bigr] f(x)\,dx,
\]
where $b(t)=b_1(t)=b_2(t)$ for $0 \leq t \leq V$. Define functions
$\tilde b_i\dvtx\mathbb{R}_+ \to\mathbb{R}$, $i=1,2$, by $\tilde b_i(t)
= b_i(V+t)$, $t \geq 0$. Then $\tilde b_1(0) = \tilde b_2(0) = b(V)$,
and
\begin{eqnarray*}
\tilde b_i(t) &=& \tilde b_i(0)
\\
&&{\fontsize{10.1pt} {12pt}\selectfont{\mbox{$\displaystyle {}+ \int
_0^t \biggl(\frac{g'(s+V)-
\int_{\mathbb{R}} \mathbb{E} [\psi^2(x+B_s-\tilde b_i(s))
e^{-\int_0^s \psi(x+B_r-\tilde b_i(r))\,dr} ] \tilde f(x)\,dx
}{\int_{\mathbb{R}} \mathbb{E} [\psi_x(x+B_s-\tilde b_i(s))
e^{-\int_0^s \psi(x+B_r-\tilde b_i(r))\,dr} ] \tilde f(x)\,dx}$}}}
\\[4pt]
&&\hspace*{50pt} {\fontsize{10.1pt} {12pt}\selectfont{\mbox{$\displaystyle{} +
\frac{1/2 \int_{\mathbb{R}} \mathbb{E} [\psi
_{xx}(x+B_s-\tilde b_i(s))
e^{-\int_0^s \psi(x+B_r-\tilde b_i(r))\,dr} ] \tilde f(x)\,dx
} {\int_{\mathbb{R}} \mathbb{E} [\psi_x(x+B_s-\tilde b_i(s))
e^{-\int_0^s \psi(x+B_r-\tilde b_i(r))\,dr} ] \tilde f(x)\,dx
} \biggr)\,ds.$}}}
\end{eqnarray*}

Fix $\varepsilon> 0$, and set
\[
K:= \min_{i=1,2} \inf_{0 \leq s \leq \varepsilon} \int
_{\mathbb{R}} \mathbb{E} \bigl[\psi_x
\bigl(x+B_s- \tilde b_i(s) \bigr) e^{-\int_0^s \psi(x+B_r-\tilde b_i(r))\,dr} \bigr]
\tilde f(x)\,dx > 0.
\]
By the triangle inequality, for $0 \leq t \leq \varepsilon$,
\[
\bigl|\tilde b_1(t) - \tilde b_2(t) \bigr| \leq \mathrm{I} +
\mathrm{II} + \mathrm{III},
\]
where
\begin{eqnarray*}
\mathrm{I} &:=& K^{-2} \int_0^t
\bigl|g'(s+V) \bigr| \int_{\mathbb{R}} \mathbb{E} \bigl[ \bigl|
\psi_x \bigl(x+B_s-\tilde b_2(s) \bigr)
e^{-\int_0^s \psi(x+B_r-\tilde b_2(r))\,dr}
\\[1pt]
&&\hspace*{113pt} {}- \psi_x \bigl(x+B_s-\tilde
b_1(s) \bigr) e^{-\int_0^s \psi(x+B_r-\tilde b_1(r))\,dr} \bigr| \bigr]
\\[2pt]
&&\hspace*{100pt} {}\times \tilde f(x)\,dx\,ds,
\\
\mathrm{II} &:=& K^{-2} \int_0^t
\int_{\mathbb{R}} \mathbb{E} \bigl[\psi^2
\bigl(x+B_s- \tilde b_1(s) \bigr) e^{-\int_0^s \psi(x+B_r-\tilde b_1(r))\,dr} \bigr]
\tilde f(x)\,dx
\\
&&{}\times\int_{\mathbb{R}} \mathbb{E} \bigl[ \bigl| \psi_x
\bigl(x+B_s-\tilde b_2(s) \bigr) e^{-\int_0^s \psi(x+B_r-\tilde b_2(r))\,dr}
\\
&&\hspace*{38pt} {}- \psi_x \bigl(x+B_s-\tilde
b_1(s) \bigr) e^{-\int_0^s \psi(x+B_r-\tilde b_1(r))\,dr} \bigr| \bigr] \tilde f(x)\,dx\,ds
\\
&&{}+ K^{-2} \int_0^t \int
_{\mathbb{R}} \mathbb{E} \bigl[ \bigl| \psi^2
\bigl(x+B_s-\tilde b_1(s) \bigr) e^{-\int_0^s \psi(x+B_r-\tilde
b_1(r))\,dr}
\\
&&\hspace*{76pt} {}- \psi^2 \bigl(x+B_s-\tilde
b_2(s) \bigr) e^{-\int_0^s \psi(x+B_r-\tilde b_2(r))\,dr} \bigr| \bigr] \tilde f(x)\,dx
\\
&&\hspace*{11pt}{}\times\int_{\mathbb{R}} \mathbb{E} \bigl[ \bigl|
\psi_x \bigl(x+B_s-\tilde b_1(s) \bigr)
e^{-\int_0^s \psi(x+B_r-\tilde b_1(r))\,dr} \bigr| \bigr] \tilde f(x)\,dx\,ds
\end{eqnarray*}
and
\begin{eqnarray*}
\mathrm{III} &:=& \frac{1}{2} K^{-2} \int_0^t
\int_{\mathbb{R}} \mathbb{E} \bigl[ \bigl|\psi_{xx}
\bigl(x+B_s-\tilde b_1(s) \bigr) e^{-\int_0^s \psi(x+B_r-\tilde
b_1(r))\,dr} \bigr|\bigr] \tilde f(x)\,dx
\\
&&{}\times\int_{\mathbb{R}} \mathbb{E} \bigl[ \bigl| \psi_x
\bigl(x+B_s-\tilde b_2(s) \bigr) e^{-\int_0^s \psi(x+B_r-\tilde b_2(r))\,dr}
\\
&&\hspace*{39pt} {} - \psi_x \bigl(x+B_s-\tilde
b_1(s) \bigr) e^{-\int_0^s \psi(x+B_r-\tilde b_1(r))\,dr} \bigr| \bigr] \tilde f(x)\,dx\,ds
\\
&&{}+ \frac{1}{2} K^{-2} \int_0^t
\int_{\mathbb{R}} \mathbb{E} \bigl[ \bigl| \psi_{xx}
\bigl(x+B_s-\tilde b_1(s) \bigr) e^{-\int_0^s \psi(x+B_r-\tilde
b_1(r))\,dr}
\\
&&\hspace*{83pt} {}- \psi_{xx} \bigl(x+B_s-\tilde
b_2(s) \bigr) e^{-\int_0^s \psi(x+B_r-\tilde b_2(r))\,dr} \bigr| \bigr] \tilde f(x)\,dx
\\
&&\hspace*{11pt}{}\times\int_{\mathbb{R}} \mathbb{E} \bigl[ \bigl|
\psi_x \bigl(x+B_s-\tilde b_1(s) \bigr) \bigr|
e^{-\int_0^s \psi(x+B_r-\tilde b_1(r))\,dr} \bigr] \tilde f(x)\,dx\,ds.
\end{eqnarray*}

Consider the integrand in $\mathrm{I}$. Note that
\begin{eqnarray*}
&& \bigl| \psi_x \bigl(x+B_s-\tilde b_2(s) \bigr)
e^{-\int_0^s \psi(x+B_r-\tilde b_2(r))\,dr}
\\
&&\quad{} - \psi_x \bigl(x+B_s-\tilde b_1(s)
\bigr) e^{-\int_0^s \psi(x+B_r-\tilde
b_1(r))\,dr} \bigr|
\\
&&\qquad\leq \bigl|\psi_x \bigl(x+B_s-\tilde
b_2(s) \bigr)\bigr|\bigl| e^{-\int_0^s \psi(x+B_r-\tilde b_2(r))\,dr} - e^{-\int_0^s
\psi(x+B_r-\tilde b_1(r))\,dr} \bigr|
\\
&&\qquad\quad{}+ e^{-\int_0^s \psi(x+B_r-\tilde b_1(r))\,dr} \bigl|\psi_x \bigl(x+B_s-
\tilde b_2(s) \bigr) - \psi_x \bigl(x+B_s-
\tilde b_1(s) \bigr) \bigr|
\\
&&\qquad\leq \|\psi_x\|_{L^\infty} s \|\psi_x\|_{L^\infty} \sup_{0 \leq r \leq s} \bigl|b_2(r) -
b_1(r) \bigr| + \|\psi_{xx}\|_{L^\infty} \sup_{0 \leq r \leq s} \bigl|b_2(r) - b_1(r)\bigr|.
\end{eqnarray*}

Similar arguments for the integrands in $\mathrm{II}$ and
$\mathrm{III}$ using the boundedness and global Lipschitz properties of
$\psi$, $\psi_x$ and $\psi_{xx}$ establish that, for a suitable
constant~$C$,
\[
\sup_{0 \leq s \leq t} \bigl|\tilde b_1(s) - \tilde
b_2(s) \bigr| \leq C \int_0^t \sup
_{0 \leq r \leq s} \bigl|\tilde b_1(r) - \tilde b_2(r) \bigr|
\,ds
\]
for $0 \leq t \leq \varepsilon$. It follows from Gr\"onwall's
inequality that $\tilde b_1(t) = \tilde b_2(t)$ for \mbox{$0 \leq t
\leq \varepsilon$}, and so $b_1(t) = b_2(t)$ for $0 \leq t \leq V +
\varepsilon$, contrary to the definition of~$V$ and the assumption that
$V$ is finite.
\end{pf}

%
%
%le2.4 #&#
\begin{lemma}[(Global existence)] \label{lglobalexistence}
Define $S$ to be the supremum of the set of~$T$ such that equations
(\ref{e1}), (\ref{e2}) and (\ref{e3}) have a continuous solution on
$[0,T]$. Then $S=+\infty$.
\end{lemma}

\begin{pf}
Suppose to the contrary that $S<+\infty$. From Lemma~\ref{luniqueness},
the equations have a unique solution on $[0,S)$. By time-reversal,
%=
%where $W$ is a Brownian motion started according to Lebesgue measure.
%Let $Z_s = W_{t-s}$ to get that
equation (\ref{e1}) is equivalent to
%G(t)&=&\E\left[\exp\left(-\int_0^t \psi(Z_{t-u} - b(u))\,du\right)
%f(Z_t)\right]\nonumber\\
%&=&\int\E\left[\exp\left(-\int_0^t \psi(x + B_{t-u} - b(u))\,du\right)
%f(x + B_t)\right]\,dx.
%
%
%e2.7 ###
%
%e2.8 #&#
\begin{equation}
\label{e*} G(t) = \int_{\mathbb{R}} \mathbb{E} \biggl[\exp \biggl(-
\int_0^t \psi \bigl(x + B_{t-u} - b(u)
\bigr)\,du \biggr) f(x + B_t) \biggr]\,dx.
\end{equation}
Similarly, (\ref{e2}) is equivalent to
%
%
%e2.9 #&#
\begin{eqnarray}
\label{e**} -g(t)&=&\int_{\mathbb{R}} \mathbb{E} \biggl[\exp \biggl(-
\int_0^t \psi \bigl(x + B_{t-u} - b(u)
\bigr)\,du \biggr)
\nonumber
\\[-8pt]
\\[-8pt]
&&\hspace*{71pt} {}\times \psi \bigl(x - b(t) \bigr) f(x + B_t)
\biggr]\,dx.
\nonumber
\end{eqnarray}

For $0 \leq t < S$ put
%
%
%e2.8 ###
%
%e2.10 #&#
\begin{equation}
\label{edefineu} u(x,t):= \mathbb{E} \biggl[\exp \biggl(-\int
_0^t \psi \bigl(x + B_{t-u} - b(u)
\bigr)\,du \biggr) f(x + B_t) \biggr].
\end{equation}

Consider $t_1 < t_2 < \cdots\uparrow S$. It follows from the continuity
of the sample paths of~$B$ that as $t_n\uparrow S$
\begin{eqnarray*}
&& \exp \biggl(-\int_0^{t_n} \psi \bigl(x +
B_{t_n-u} - b(u) \bigr)\,du \biggr) f(x + B_{t_n})
\\
&&\qquad\rightarrow\exp \biggl(-\int_0^S \psi
\bigl(x + B_{S-u} - b(u) \bigr)\,du \biggr) f(x + B_S)
\end{eqnarray*}
almost surely for each $x \in\mathbb{R}$, and so
\[
u(x,t_n) \rightarrow\mathbb{E} \biggl[\exp \biggl(-\int
_0^S \psi \bigl(x + B_{S-u} - b(u)
\bigr)\,du \biggr) f(x + B_S) \biggr] =: u(x,S).
\]
Because
\[
u(x,t) \leq \mathbb{E} \bigl[f(x + B_t) \bigr],
\]
it follows from dominated convergence that
\[
\int_{\mathbb{R}} u(x,S)\,dx = \lim_n \int
_{\mathbb{R}} u(x,t_n)\,dx = \lim_n
G(t_n) = G(S).
\]
Also,
\[
\lim_n \int_{\mathbb{R}} \psi \bigl(x -
b(t_n) \bigr) u(x,t_n)\,dx = - \lim
_n g(t_n) = -g(S).
\]
Because $0 < -g(S) < G(S)$ and
\[
u(x,S) \geq e^{-S} \mathbb{E} \bigl[f(x + B_S) \bigr] > 0,
\qquad x \in\mathbb{R},
\]
there is, by Lemma~\ref{lIVP}, a unique $b^* \in\mathbb{R}$ such that
\[
\int_{\mathbb{R}} \psi \bigl(x - b^* \bigr) u(x,S)\,dx = -g(t).
\]

We claim that $b(t_n) \rightarrow b^*$. If this was not the case, then,
by passing to a subsequence we would have $b(t_n)$ converging to
some other extended real $c$ and hence, by dominated convergence,
\begin{eqnarray*}
-g(t)&=& - \lim_n g(t_n)
\\
&=& \lim_n \int_{\mathbb{R}} \psi \bigl(x -
b(t_n) \bigr) u(x,t_n)\,dx
\\
&=& \int_{\mathbb{R}} \psi(x - c) u(x,S)\,dx,
\end{eqnarray*}
contradicting the definition of $b^*$ [where we used the natural
definitions $\psi(-\infty):=1$, $\psi(+\infty):=0$].
Using dominated convergence in (\ref{e3}) we get that there exists a
continuous $b$ such that all three equations hold on $[0,S]$.\vadjust{\goodbreak}

All we need to do now is show that we can extend the existence from
$[0,S]$ to $[0,S+\delta]$ for some $\delta>0$. This amounts to
proving existence on $[0,\delta]$ starting at a different initial
condition---replacing the
original probability density $f$ by the density of the probability measure
\[
\int_{\mathbb{R}} \mathbb{E} \biggl[\exp \biggl(-\int
_0^S \psi \bigl(x + B_u - b(u)
\bigr)\,du \biggr), B_S \in\bullet \biggr]f(x)\,dx / G(S).
\]
This will follow if we can establish the local existence, that is, the
existence for some $\delta>0$, of a solution of the following PDE/ODE system:
\[
\cases{\displaystyle\tilde u_t(x,t) = \frac{1}{2}\tilde
u_{xx}(x,t) - \psi \bigl(x-\tilde{b}(t) \bigr)\tilde u(x,t),\qquad x\in
\mathbb{R}, 0 < t < \delta, \vspace*{2pt}
\cr
\tilde u(x,0) =u(x,S)/G(S),
\hspace*{106pt} x \in\mathbb{R}, \vspace*{2pt}
\cr
\tilde{b}(0) =b(S),
\vspace*{2pt}
\cr
{\fontsize{10.5pt} {12pt}\selectfont{\mbox{$\displaystyle
\tilde{b}'(t) = \frac{(g(S+t)+g'(S+t))/G(S) - \int_{\mathbb{R}}
[\psi^2(x-\tilde{b}(t))-\psi(x-\tilde{b}(t))] \tilde u(x,t)\,dx
}{\int_{\mathbb{R}} \psi_x(x-\tilde{b}(t)) \tilde u(x,t)\,dx}$}}} \vspace*{2pt}
\cr
{
\fontsize{10.5pt} {12pt}\selectfont{\mbox{$\displaystyle\phantom{
\tilde{b}'(t) =} {}-\frac{1/2 \int_{\mathbb{R}} \psi
_x(x-\tilde{b}(t)) \tilde
u_x(x,t)\,dx}{\int_{\mathbb{R}} \psi_x(x-\tilde{b}(t)) \tilde u(x,t)\,dx},\hspace*{36pt} 0<t<
\delta.$}}}}
\]
We note that the expression for $\tilde{b}'(t)$ is not the analogue of
the one for $b'(t)$ that arises immediately from differentiating
(\ref{e3}), which in turn arose from rearranging~(\ref{eg}) and
integrating. However, adding
$0=\int_{\mathbb{R}}\psi(x-b(t))u(x,\break t)\,dx - g(t)$ to the right-hand
side of (\ref{eg}) and then solving for $b'(t)$ leads to an expression
of this form. Note that
\[
u(x,S)=\mathbb{E} \biggl[\exp \biggl(-\int_0^S
\psi \bigl(x + B_{S-u} - b(u) \bigr)\,du \biggr) f(x + B_S)
\biggr] > 0
\]
and, by dominated convergence, that $u(\cdot,S)\in C^2(\mathbb{R})$
with $\|u(\cdot,S)\|_{L^\infty(\mathbb{R})}$,
$\|u_x(\cdot,S)\|_{L^\infty(\mathbb{R})}$,
$\|u_{xx}(\cdot,S)\|_{L^\infty(\mathbb{R})}$ all finite. Therefore, we
can apply Theorem~\ref{texistence} below to get that there is a time
$\delta>0$ and a unique pair $\tilde u,\tilde{b}$ satisfying the
PDE/ODE system above with $\tilde u$ twice continuously differentiable
in $x$ on $\mathbb{R}$ and once continuously differentiable in $t$ on
$[0,\delta]$, that is, $\tilde u\in
C_x^2(\mathbb{R})C_t^1([0,\delta])$, and with $\tilde{b}\in
C^1([0,\delta])$. Thus, we have proven that we have a unique continuous
$b$ satisfying equations (\ref{e1}), (\ref{e2}) and (\ref{e3}) on
$[0,S+\delta]$. This contradicts the maximality of $S$. As a result,
$S=\infty$ and we are done.
\end{pf}

This completes the proof of Theorem~\ref{tglobalEandU}.
\end{pf}

%
%
%re2.5 #&#
\begin{remark}\label{rlocal}
Theorem \ref{texistence} below gives local in time existence and
uniqueness of solutions to the system (\ref{ePDE}). However, we require
the global uniqueness result Lemma \ref{luniqueness} because it is not
a priori clear that all the solutions to equations
(\ref{e1})--(\ref{e3}) are solutions to the system (\ref{ePDE}).
\end{remark}

%
%
%re2.6 #&#
\begin{remark}\label{rrightderivative}
It follows from the (\ref{e3}), the smoothness assumptions on $G$, the
smoothness assumptions on $\psi$, the smoothness assumptions on $f$ and
the assumption that $f$ is everywhere positive that the function $b$
has a finite right derivative at $0$. In the standard inverse first
passage problem, the analogous property for the boundary often fails
(e.g., when the lifetime distribution is exponential).
\end{remark}

As a corollary we get the global existence and uniqueness of the
PDE/ODE system.

%
%
%co2.7 #&#
\begin{corollary}
Suppose that the conditions of Theorem~\ref{tglobalEandU} hold. Then
the system
%
%
%e2.9 ###
%
%e2.11 #&#
\begin{equation}
\label{ePDEODE} \qquad\cases{\displaystyle u_t(x,t) =
\frac{1}{2}u_{xx}(x,t) - \psi \bigl(x-b(t) \bigr)u(x,t),
\vspace*{2pt}
\cr
u(x,0) =f(x),\hspace*{140pt} x \in \mathbb{R}, \vspace*{2pt}
\cr
\displaystyle-g(0)= \int_\mathbb{R} \psi \bigl(x-b(0) \bigr)
f(x)\,dx, \vspace*{2pt}
\cr
\displaystyle b'(t) =
\frac{g(t)+g'(t) - \int_{\mathbb{R}} [\psi^2(x-b(t))-\psi
(x-b(t))] u(x,t)\,dx }{\int_{\mathbb{R}} \psi_x(x-b(t))u(x,t)\,dx} \vspace*{2pt}
\cr
\displaystyle\phantom{b'(t) =}
{} -\frac{ 1/2 \int_{\mathbb{R}} \psi_x(x-b(t)) u_x(x,t)\,
dx}{\int_{\mathbb{R}} \psi_x(x-b(t))u(x,t)\,dx},\qquad t > 0,}
\end{equation}
has a unique solution $(u,b)\in C_x^2(\mathbb{R})C^1_t(\mathbb
{R}_+)\times C^1_t(\mathbb{R}_+)$.
\end{corollary}

%s3 ###
%s3 #&#
\section{Local existence and uniqueness}

We now consider the PDE/ODE system (\ref{ePDEODE}). We have already
used the standard notation $F_x$ and $F_{xx}$ to denote the first and
second derivatives of a function $F$ of one variable or the first and
second partial derivatives with respect to the variable $x$ of a
function $F$ of several variables. Because we repeatedly deal with the
function $(x,t) \mapsto\psi(x - b(t))$, it will be convenient to
recycle notation and define a function $\psi_{b}\dvtx\mathbb {R}\times
\mathbb{R}_+ \to\mathbb{R}$ by $\psi_b(x,t) =\psi(x-b(t))$. We will
then set $\psi_{x,b}:=\partial_x\psi_b$ and
$\psi_{xx,b}:=\partial_{xx}\psi_b$. We will continue to use the
notation $\psi_x$ and $\psi_{xx}$ with its old meaning, but there
should be no confusion between the different objects $\psi_b$ and
$\psi_x$. Similarly, we set $\phi:=\psi^2-\psi= -\psi(1 - \psi)$ and
put $\phi_{b}(x,t)=\phi(x-b(t))$. Finally, for two functions $f$, $g$
and fixed $t\geq0$ define $\langle f$,
$g\rangle=\int_{\mathbb{R}}f(x,t)g(x,t)\,dx$.

%Adding and subtracting $\int_{\mathbb{R}}\psi(x-b(t))u(x,t)\,dx$ from
%the above equation \eqref{eg} and then solving for $b'(t)$
%b'(t) = \frac{g(t) + g'(t) - \la\phi_b,u\ra- 1/2 \la\psi_{x,b},u_x
In the notation we have introduced, we wish to consider the system
%
%
%e3.1 ###
%
%e3.1 #&#
\begin{equation}\label{ePDEsmooth}
\qquad\cases{\displaystyle u_t(x,t) =
\frac{1}{2}u_{xx}(x,t) - \psi \bigl(x-b(t) \bigr)u(x,t), &\quad
$x \in\mathbb{R}$, $t > 0$, \vspace*{2pt}
\cr
u(x,0) =f(x), &\quad $x \in
\mathbb{R}$, \vspace*{2pt}
\cr
b(0) = b_0, \vspace*{2pt}
\cr
\displaystyle b'(t) = \frac{g(t) + g'(t) - \langle\phi_b,u\rangle- 1/2 \langle
\psi_{x,b},u_x\rangle}{\langle\psi_{x,b},u\rangle}, &\quad $t > 0$,}
\end{equation}
for some $b_0 \in\mathbb{R}$. [In the proof of
Theorem~\ref{tglobalEandU} we choose $b_0$ to satisfy $-g(0)=
\int_\mathbb{R}\psi(x-b_0) f(x)\,dx$, but we may take an arbitrary
value for $b_0$ and still obtain a local existence and uniqueness
result.]

We have assumed in the statement of Theorem~\ref{tglobalEandU} that
$f\in C^2(\mathbb{R})$ and $\psi\in C^3(\mathbb{R})$ with $\|\psi\|
_{L^\infty}=1$, $\|\psi\|_{L^\infty}=:B$, $\|\psi_{xx}\|_{L^\infty
}=:C$, and $\|\psi_{xxx}\|_{L^\infty}=:F$ for finite constants $B$,
$C$, $F$. Furthermore, we have assumed for some $h>0$ that $\psi(x)=1$
for $x\leq-h$, that $\psi(x)=0$ for $x\geq h$ and that $\psi\geq0$ and
$\psi_x\leq0$ for all $x \in\mathbb{R}$. Set
$\int_{\mathbb{R}}|\psi_x(x)|\,dx=:D$, and note that $0<D<\infty$. It
is immediate that $\|\phi\|_{L^\infty}\leq1$ and $\|\phi_x\|
_{L^\infty} = \|\psi_x(1 - 2 \psi)\|_{L^\infty}
\leq\|\psi_x\|_{L^\infty}=B$. Moreover, the functions $\phi$ and
$\phi_x$ are supported on $[-h,h]$ and
\mbox{$0<\int_{\mathbb{R}}|\phi(x)|\,dx=:E<\infty$}.
%$0<\int_{\mathbb{R}}|\phi_x(x)|\,dx=F<\infty$.

%
%
%de3.1 #&#
\begin{definition}
For $T>0$, let $(\mathcal{L}^T, \|\cdot\|_T)$
be the Banach space consisting of pairs of functions
$(u,b)$ such that $u\in C^2_x(\mathbb{R})C_t([0,T])$, $b\in C([0,T])$ and
%
%
%
%e3.2 #&#
\begin{eqnarray}\label{enorm}
\bigl\| (u,b)\bigr\| _T &:=& \|u\|_{L_x^\infty(\mathbb{R}) L_t^\infty([0,T]) }\nonumber
\\[-1pt]
&&{}+ \|u_x\|_{L_x^\infty(\mathbb{R}) L_t^\infty([0,T])} + \|u_{xx}\|_{L_x^\infty(\mathbb{R}) L_t^\infty([0,T]) }
\nonumber\\[-9pt]\\[-9pt]
&&{}+ \|b\|_{L^\infty([0,T])}\nonumber
\\[-1pt]
&<& \infty.\nonumber
\end{eqnarray}
\end{definition}

%
%
%de3.2 #&#
\begin{definition}
\label{DGamma} Given constants $M$, $N$, $P$, $A$, $L > 0$, $b_0
\in\mathbb{R}$ and $T>0$, define the closed subset $\Gamma^T_{\mathit{MNPALb}_0}
\subset\mathcal{L}^T$ by
%
%
%e3.3 #&#
\begin{eqnarray}\label{espace}
\Gamma^T_{\mathit{MNPALb}_0}&:=& \Bigl\{ (u,b) \in \mathcal{L}^T\dvtx\nonumber
\\[-1pt]
&&\hspace*{5pt}\|u\|_{L_x^\infty L_t^\infty([0,T])}\leq M,\nonumber
\\[-1pt]
&&\hspace*{5pt}\|u_x\|_{L_t^\infty([0,T]) L_x^\infty}\leq N,\nonumber
\\[-1pt]
&&\hspace*{5pt}\|u_{xx}\|_{L_t^\infty([0,T]) L_x^\infty}\leq P,
\\[-1pt]
&&\hspace*{5pt} b(0) = b_0,\nonumber
\\[-1pt]
&&\hspace*{5pt}\|b\|_{L^\infty([0,T])}\leq A/2,\nonumber
\\[-1pt]
&&\hspace*{5pt}\inf_{x\in[-A,A], t\in[0,T]} u(x,t)\geq L \Bigr\}.\nonumber
\end{eqnarray}
\end{definition}

The following is the main result of this section.

%
%
%th3.3 #&#
\begin{theorem}\label{tcontraction}
Suppose that the assumptions of Theorem~\ref{tglobalEandU} hold.
Suppose also that the constants $M$, $N$, $P$, $A$, $L > 0$ and $b_0
\in\mathbb {R}$ are such that:
\begin{itemize}
\item $|b_0| \leq A/4$,\vspace*{2pt}

\item $f(x)\geq4L>0$ for $x\in[-A,A]$,\vspace*{2pt}

\item $\|f\|_{L^\infty(\mathbb{R})}\leq M/2$,\vadjust{\goodbreak}%\vspace*{2pt}

\item $\|f_x\|_{L^\infty(\mathbb{R})}\leq N/2$,\vspace*{2pt}

\item $\|f_{xx}\|_{L^\infty(\mathbb{R})}\leq P/2$.
\end{itemize}
Then for $T>0$ sufficiently small, there is a contractive map
$\Phi:\Gamma^T_{\mathit{MNPALb}_0}\rightarrow\Gamma^T_{\mathit{MNPALb}_0}$ defined by
$\Phi(v,b) = (u,c)$,
where
%
%
%e3.15 ###
%
%e3.4 #&#
\begin{equation}\label{ePDEuv}
\qquad\cases{\displaystyle u_t(x,t) = \frac{1}{2}u_{xx}(x,t)
- \psi \bigl(x-b(t) \bigr)v(x,t), &\quad $x \in\mathbb{R}$, $t>0$, \vspace*{2pt}
\cr
u(x,0) =f(x), &\quad $x \in\mathbb{R}$, \vspace*{2pt}
\cr
\displaystyle
c'(t) = \frac{g(t) + g'(t) - \langle\phi_b,v\rangle- 1/2 \langle
\psi_{x,b},v_x\rangle}{\langle\psi_{x,b},v\rangle}, &\quad $0 < t \leq T$, \vspace*{2pt}
\cr
c(0) = b_0.}
\end{equation}
\end{theorem}

We will prove Theorem~\ref{tcontraction} in a series of lemmas. Each
lemma will assume the hypotheses of Theorem~\ref{tcontraction} and the
bounds established in the previous lemmas.
%, that is we assume that the time is already small enough so that the
%previous Lemmas hold. The minimums of a finite number of strictly
%positive quantities will still be positive so this process always
%gives us a $T>0$.
%Thus, in each Lemma we will use the bounds we have from the previous
%ones.

%
%
%re3.4 #&#
\begin{remark}
Since $f$ is continuous and positive, for any $A>0$ there exists $L>0$
such that $f(x)\geq4L$ for $x \in[-A,A]$. Therefore, we are not
restricting the possible values of $b(0)$ by the above assumptions. We
will also assume without loss of generality that $h\leq A/4$.
\end{remark}

%
%
%le3.5 #&#
\begin{lemma}[(Boundedness of $u$)]\label{lboundu} Suppose that
$(u,c)=\Phi((v,b))$, with $(v,b) \in\Gamma^T_{\mathit{MNPALb}_0}$.
Then, there exists a time $T>0$ such that
\[
\|u\|_{L_x^\infty L_t^\infty([0,T])}\leq M.
\]
\end{lemma}

\begin{pf}
Using Duhamel's formula [see (\ref{eduhamel})],
\begin{eqnarray*}
\bigl|u(x,t)\bigr| &=& \Biggl|\int_{\mathbb{R}}G(y,t)f(x-y)\,dy
\\
&&\hspace*{2pt}{} -\int_0^t\int_{\mathbb{R}}G(x-y,t-s)
\psi_{c(s)}(y) v(y,s)\,dy\,ds \Biggr|
\\
&\leq& \int_{\mathbb{R}}G(y,t)f(x-y)\,dy
\\
&&{} + \int_0^t \int_{\mathbb {R}}G(x-y,t-s)\bigl|\psi_{c(s)}(y)\bigr| \bigl|v(y,s)\bigr|\,dy\,ds
\\
&\leq& M/2 \int_{\mathbb{R}}G(y,t)\,dy + M \int_0^t
\int_{\mathbb
{R}}G(x-y,t-s)\,dy\,ds
\\
&\leq& M/2+Mt
\\
&\leq& M
\end{eqnarray*}
when $t \leq \frac{1}{2}$, where
\[
G(x,t):=\frac{1}{\sqrt{2\pi t}}e^{-x^2/2t}, \qquad x \in \mathbb{R}, t > 0.
\]\upqed
\end{pf}

%
%
%le3.6 #&#
\begin{lemma}[(Boundedness of $u_x$)]\label{lboundux}
Suppose that $(u,c)=\Phi((v,b))$ with $(v,b) \in\Gamma
^T_{\mathit{MNPALb}_0}$. Then there exists a time $T>0$ such that
\[
\|u_x\|_{L_t^\infty([0,T]) L_x^\infty}\leq N.
\]
\end{lemma}

\begin{pf}
Since $u_x$ solves
\[
\cases{\displaystyle \biggl(\partial_t-\frac{\partial_{xx}}{2}
\biggr)u_x =-\psi_{x,c}v-\psi_c v_x,
&\quad $x\in\mathbb{R}$, $t>0$, \vspace*{2pt}
\cr
u_x(x,0) =
f_x(x),}
\]
we have via Duhamel's formula that
\begin{eqnarray*}
\bigl|u_x(x,t)\bigr| &=& \Biggl|\int_{\mathbb{R}}G(y,t)f_x(x-y) \,dy
\\
&&\hspace*{3pt}{}+\int_0^t\int_{\mathbb{R}}G(x-y,t-s)
(-\psi_{x,c} v-\psi_{c}v_x) (y,s)\,dy\,ds \Biggr|
\\
&\leq& \int_{\mathbb{R}}G(y,t)\bigl|f_x(x-y)\bigr|\,dy
\\
&&{} + \int _0^t\int_{\mathbb{R}}G(x-y,t-s)|\psi_{x,c}| \bigl|v(y,s)\bigr|\,dy\,ds
\\
&&{}+ \int_0^t\int_{\mathbb{R}}G(x-y,t-s)\bigl|\psi \bigl(y-c(s) \bigr) \bigr|\bigl|v_x(y,s)\bigr|\,dy\,ds
\\
&\leq& \frac{N}{2} + MB \int_0^t\int_{\mathbb{R}} G(x-y,t-s)\,dy\,ds
\\
&&{}+ N \int_0^t \int_{\mathbb{R}}G(x-y,t-s)\,dy\,ds
\\
&\leq& \frac{N}{2} + MB t + N t.
\end{eqnarray*}
Thus
\[
\|u_x\|_{L_t^\infty([0,T]) L_x^\infty} \leq\frac{N}{2} + (MB+ N) T \leq N,
\]
whenever $T\leq T^\ast$, where
\[
T^\ast=\frac{N}{2(MB+N)}.
\]\upqed
\end{pf}

%
%
%le3.7 #&#
\begin{lemma}[(Boundedness of $u_{xx}$)]\label{lbounduxx}
Suppose that $(u,c)=\Phi((v,b))$ with $(v,b) \in\Gamma^T_{\mathit{MNPALb}_0}$.
Then, there exists a time $T>0$ such that
\[
\|u_{xx}\|_{L_t^\infty([0,T]) L_x^\infty}\leq P.
\]
\end{lemma}

\begin{pf}
Noting that $u_{xx}$ solves
\[
\cases{\displaystyle \biggl(\partial_t-\frac{\partial_{xx}}{2}
\biggr)u_{xx} =-\psi_{xx,c}v-2\psi_{x,c}v_x-
\psi_c v_{xx}, &\quad $x\in\mathbb{R}$, $t>0$,
\vspace*{2pt}
\cr
u_{xx}(x,0) = f_{xx}(x),}
\]
analogous manipulations to those from Lemma \ref{lboundux} yield the
result.
\end{pf}
%
%
%le3.8 #&#
\begin{lemma}[(Lower bound for $u$ and boundedness of $c'$ and $c$)]\label{plowerbound}
Suppose that $(u,c)=\Phi((v,b))$ with $(v,b) \in\Gamma^T_{\mathit{MNPALb}_0}$.
Then, there exists a time $T>0$ such that
%
%
%e3.16 ###
%
%e3.5 #&#
\begin{equation}
\label{euniformbounds} u \geq L\qquad\mbox{on }x\in[-A,A], t\in[0,T]
\end{equation}
and
$c(t)\in[-A/2,A/2]$ for $t\in[0,T]$.
\end{lemma}

\begin{pf}
Recall that $b(0)\in[-A/4,A/4]$.
Then it is immediate that
%
%
%e3.17 ###
%
%e3.6 #&#
%e3.7 #&#
\begin{eqnarray}\label{elowerbound}
\Biggl|\int_{\mathbb{R}} \psi_x
\bigl(x-b(t) \bigr)v(x,t)\,dx \Biggr|= \Biggl|\int_{\mathbb{R}}
\psi_x(y)v \bigl(y+b(t) \bigr)\,dy \Biggr| \geq DL,
\nonumber\\[-12pt]\\[-8pt]
\eqntext{t\in[0,T],}
\end{eqnarray}
because on the support $[-h,h]$ of $\psi_x$ we have $y\in
[-h,h]\subseteq[-A/4,A/4]$ which together with the bound on $b(t)$
implies $y+b(t)\in[-A,A]$. Therefore, $v(y+b(t))\geq L$ for $t\in
[0,T]$ which, since $\psi_x\leq0$, yields
\[
\int_{\mathbb{R}}\psi_x(y)v \bigl(y+b(t) \bigr)\,dy \leq
L \int_{\mathbb
{R}}\psi_x(y)\,dy = -LD<0, \qquad t
\in[0,T].
\]
We see from these bounds that
\[
\bigl|c'(t)\bigr|\leq\frac{\sup_{[0,t]}(|g+g'|)+ME+ND/2}{LD}
\]
and, by integrating,
\[
\bigl|c(t)\bigr|\leq\bigl|c(0)\bigr| + \frac{\sup_{[0,t]}(|g+g'|)+ME+ND/2}{LD}t.
\]
Thus, there is $T>0$ such that for $t\in[0,T]$,
\[
\bigl|c(t)\bigr|\in[-A/2,A/2].
\]
Using the assumptions, equation (\ref{eduhamel}) gives
\begin{eqnarray*}
u(x,t) &=& \int_{\mathbb{R}}G(y,t)f(x-y)\,dy -\int
_0^t\int_{\mathbb
{R}}G(x-y,t-s)
\psi_{c(s)}(y) v(y,s)\,dy\,ds
\\
&\geq& 4 L \int_{x-A}^{x+A} G(y,t)\,dy - M \int
_0^t\int_{\mathbb
{R}}G(x-y,t-s)\,dy
\,ds
\\
&\geq& 4 L\int_{x-A}^{x+A} G(y,t)\,dy - Mt.
\end{eqnarray*}
If $0\leq x\leq A$, then $x-A\leq0$ and $x+A\geq A>0$, so for small
enough $t$ we have
\[
\int_{x-A}^{x+A} G(y,t)\,dy \geq\int
_{0}^{A} G(y,t)\,dy \geq\frac{1}{3}.
\]
If $ -A \leq x<0$, then $x+A\geq0$ and $x-A\leq-A <0$.
So, for small enough $t$,
\[
\int_{x-A}^{x+A} G(y,t)\,dy \geq\int
_{-A}^{0} G(y,t)\,dy \geq\frac{1}{3}.
\]
Therefore, there exists a time $T>0$ such that whenever $t\in[0,T]$
and $x\in[-A,A]$,
\[
u(x,t) \geq \tfrac{4}{3}L - Mt
\geq L.
\]\upqed
\end{pf}

%
%
%le3.9 #&#
\begin{lemma}
For a sufficiently small time $T > 0$, the set $\Gamma^T_{\mathit{MNPALb}_0}$ is
mapped into itself by~$\Phi$.
\end{lemma}

\begin{pf}
The above lemmas provided the necessary bounds. Now, note that if we
start with $(v,b) \in\Gamma^T_{\mathit{MNPALb}_0}$, then we first get the
function $c$ from the last two equations in (\ref{ePDEuv}) by simply
integrating. The integration is well defined because the denominator is
bounded in absolute value below by $DL>0$ and the numerator is bounded
above. Thus $c\in C^1([0,t])$. Next, having $c$ in hand we get the
function $u$ from the first two equations of (\ref{ePDEuv}). We note
that, by Duhamel's formula, the function $u$ has actually more than the
desired smoothness, namely, $u\in C^2_x(\mathbb{R})C^1_t([0,T])$.
\end{pf}

%
%
%le3.10 #&#
\begin{lemma}\label{lcontractionb}
Suppose that $(v_1,b_1), (v_2,b_2)\in\Gamma^T_{\mathit{MNPALb}_0}$. Set
$(u_1,c_1)=\Phi((v_1,b_1))$ and $(u_2,c_2)=\Phi((v_2,b_2))$. For any
$\varepsilon>0$ there exists $T>0$ such that
%
%
%e3.18 ###
%
%e3.8 #&#
\begin{equation}
\label{econtractionb}
\|c_2-c_1\|_{L_t^\infty([0,T])} \leq\varepsilon\bigl\|(v_2, b_2) - (v_1,b_1) \bigr\|_T.
\end{equation}
\end{lemma}

\begin{pf}
Note that the functions $c_1$, $c_2$ satisfy
%e3.19 ###
%
%e3.9 #&#
\begin{equation}
\label{ebc} \cases{\displaystyle c_1'(t) =
\frac{g(t)+g'(t) - \langle\phi_{b_1},v_1\rangle- 1/2
\langle\psi_{x,b_1},\partial_xv_1\rangle}{\int_{\mathbb{R}}
\langle\psi_{x,b_1},v_1\rangle}, &\quad $t>0$, \vspace*{2pt}
\cr
\displaystyle
c_2'(t) = \frac{g(t)+g'(t) - \langle\phi_{b_2},v_2\rangle- 1/2
\langle\psi_{x,b_2},\partial_xv_2\rangle}{\int_{\mathbb{R}}
\langle\psi_{x,b_2},v_2\rangle}, &\quad$t>0$.}
\end{equation}
Subtracting the two equations gives
\begin{eqnarray*}
&& c_2'(t)-c_1'(t)
\\
&&\qquad =
\bigl[g(t)+g'(t) \bigr] \biggl(\frac{\langle\psi
_{x,b_1},v_1\rangle-\langle\psi_{x,b_1},v_2\rangle}{\langle\psi
_{x,b_1},v_1\rangle\langle\psi_{x,b_2},v_2\rangle} +
\frac{\langle\psi_{x,b_1},v_2\rangle-\langle\psi
_{x,b_2},v_2\rangle}{\langle\psi_{x,b_1},v_1\rangle\langle\psi
_{x,b_2},v_2\rangle} \biggr)
\\
&&\quad\qquad{}+ \frac{ (\langle\phi_{b_1},v_1\rangle-\langle\phi
_{b_2},v_1\rangle)\langle\psi_{x,b_2},v_2\rangle}{\langle
\psi_{x,b_1},v_1\rangle\langle\psi_{x,b_2},v_2\rangle} + \frac{
(\langle\phi_{b_2},v_1\rangle-\langle\phi
_{b_2},v_2\rangle)\langle\psi_{x,b_2},v_2\rangle}{\langle
\psi_{x,b_1},v_1\rangle\langle\psi_{x,b_2},v_2\rangle}
\\
&&\quad\qquad{}+ \frac{ (\langle\phi_{b_2},v_2\rangle-\langle\phi
_{b_1},v_2\rangle)\langle\phi_{b_2},v_2\rangle}{\langle\psi
_{x,b_1},v_1\rangle\langle\psi_{x,b_2},v_2\rangle} + \frac{
(\langle\phi_{b_1},v_2\rangle-\langle\phi
_{b_1},v_1\rangle)\langle\phi_{b_2},v_2\rangle}{\langle\psi
_{x,b_1},v_1\rangle\langle\psi_{x,b_2},v_2\rangle}
\\
&&\quad\qquad{}+ \frac{ (\langle\psi_{x,b_1},\partial_xv_1\rangle
-\langle\psi_{x,b_2},\partial_x v_1\rangle)\langle\psi
_{x,b_2},v_2\rangle}{2\langle\psi_{x,b_1},v_1\rangle\langle\psi
_{x,b_2},v_2\rangle}
\\
&&\quad\qquad{}+ \frac{ (\langle\psi_{x,b_2},\partial_xv_1\rangle
-\langle\psi_{x,b_2},\partial_xv_2\rangle)\langle\psi
_{x,b_2},v_2\rangle}{2\langle\psi_{x,b_1},v_1\rangle\langle\psi
_{x,b_2},v_2\rangle}
\\
&&\quad\qquad{}+ \frac{ (\langle\psi_{x,b_2},v_2\rangle-\langle\psi
_{x,b_1},v_2\rangle)\langle\psi_{x,b_2},\partial_xv_2\rangle
}{2\langle\psi_{x,b_1},v_1\rangle\langle\psi_{x,b_2},v_2\rangle}
\\
&&\quad\qquad{}+\frac{ (\langle\psi_{x,b_1},v_2\rangle-\langle\psi
_{x,b_1},v_1\rangle)\langle\psi_{x,b_2},\partial_xv_2\rangle
}{2\langle\psi_{x,b_1},v_1\rangle\langle\psi_{x,b_2},v_2\rangle}.
\end{eqnarray*}
Using the fact that the functions $\psi$, $\psi_x$ and $\phi$ are
Lipschitz, that $v_1$ and $v_2$ are bounded, and that their first
derivatives are bounded,
we find that
\begin{eqnarray*}
\bigl\|c_2'-c_1'\bigr\|_{L_t^\infty([0,T])}
&\leq& \frac{\sup_{[0,T]}|g+g'|\|
v_1-v_2\|_{L_x^\infty L_t^\infty([0,T])}}{L^2D^2}
\\
&&{}+ \frac{\sup_{[0,T]}|g+g'|MC(A+2h)\|b_2-b_1\|_{L_t^\infty
([0,T])}}{L^2D^2}
\\
&&{}+\frac{DM^2B(A+2h)\|b_2-b_1\|_{L_t^\infty([0,T])}}{L^2D^2}
\\
&&{}+\frac{DME\|v_2-v_1\|_{L_x^\infty L_t^\infty([0,T])}}{L^2D^2}
\\
&&{}+ \frac{EM^2B(A+2h)\|b_2-b_1\|_{L_t^\infty([0,T])}}{L^2D^2}
\\
&&{}+\frac{ME^2\|v_2-v_1\|_{L_x^\infty L_t^\infty([0,T])}}{L^2D^2}
\\
&&{}+ \frac{NMDC(A+2h)\|b_2-b_1\|_{L_t^\infty([0,T])}}{2L^2D^2}
\\
&&{}+ \frac{MD^2\|\partial_x v_2-\partial_x v_1\|_{ L_x^\infty
L_t^\infty([0,T])}}{2L^2D^2}
\\
&&{}+ \frac{NMDC(A+2h)\|b_2-b_1\|_{L_t^\infty([0,T])}}{2L^2D^2}
\\
&&{}+\frac{ND^2\|v_2-v_1\|_{L_x^\infty L_t^\infty([0,T])}}{2L^2D^2}.
\end{eqnarray*}
Integrating and recalling that $c_1(0)=c_2(0)=b_0$ leads to
\begin{eqnarray*}
\biggl|\int_0^t \bigl(c_2'(s)-c_1'(s)
\bigr)\,ds \biggr|&=& \bigl|c_2(t)-c_1(t) -
\bigl(c_2(0)-c_1(0) \bigr)\bigr|
\\
&\leq& \int_0^t \bigl|c_2'(s)-c_1'(s)\bigr|\,ds
\\
&\leq& t\bigl\|c_2'-c_1'\bigr\|_{L_t^\infty([0,t])}.
\end{eqnarray*}
Hence,
\[
\|c_2-c_1\|_{L_t^\infty([0,T])} \leq T \bigl\|c_2'-c_1'
\bigr\|_{L_t^\infty([0,T])}
\]
and by the above bound on $\|c_2'-c_1'\|_{L_t^\infty([0,T])}$ for any
$\varepsilon> 0$ we can choose $T$ small enough that
\[
\|c_2-c_1\|_{L_t^\infty([0,T])} \leq \varepsilon
\bigl\|(v_2, b_2) - (v_1, b_1)
\bigr\|_T.
\]\upqed
\end{pf}
%
%&+&\frac{DM^23B(A+2h)\|b_2-b_1\|_{L_t^\infty([0,T])}}{L^2D^2} +
%&+& \frac{EM^23B(A+2h)\|b_2-b_1\|_{L_t^\infty([0,T])}}{L^2D^2}+
%&+& \frac{NMDC(A+2h)\|b_2-b_1\|_{L_t^\infty([0,T])}}{2L^2D^2} +
%&+& \frac{NMDC(A+2h)\|b_2-b_1\|_{L_t^\infty([0,T])}}{2L^2D^2}+

%
%
%le3.11 #&#
\begin{lemma}\label{lcontriactionu}
Suppose that $(v_1,b_1), (v_2,b_2)\in\Gamma^T_{\mathit{MNPALb}_0}$. Set
$(u_1,c_1)=\Phi((v_1,b_1))$ and $(u_2,c_2)=\Phi((v_2,b_2))$. For any
$\varepsilon>0$ there exists $T>0$ such that
%
%
%e3.20 ###
%
%e3.10 #&#
\begin{equation}
\label{econtractionu} \|u_2-u_1\|_{L_x^\infty L_t^\infty([0,T])} \leq
\varepsilon\bigl\|(v_2, b_2) - (v_1,b_1) \bigr\|_T.
\end{equation}
\end{lemma}

\begin{pf}
The following equations hold:
%
%
%e3.21 ###
%
%e3.11 #&#
\begin{equation}
\label{eubv} \cases{\displaystyle \biggl(\partial_t-
\frac{\partial_{xx}}{2} \biggr)u_1 =-\psi \bigl(x-c_1(t)
\bigr)v_1, &\quad$x\in\mathbb{R}$, $t>0$, \vspace*{4pt}
\cr
\displaystyle\biggl(\partial_t-\frac{\partial_{xx}}{2} \biggr)u_2 =-\psi
\bigl(x-c_2(t) \bigr)v_2, &\quad$x\in\mathbb{R}$, $t>0$,
\vspace*{4pt}
\cr
u_1(x,0) = f(x), &\quad$x \in\mathbb{R}$,
\vspace*{3pt}
\cr
u_2(x,0) = f(x), &\quad$x \in\mathbb{R}$.}
\end{equation}
By Duhamel's formula we have
%
%
%e3.22 ###
%
%e3.12 #&#
\begin{equation}
\label{eduhamel1} u_1 = G\ast(f\delta_{t=0}) + G\ast(-
\psi_{c_1}v_1)
\end{equation}
and
%
%
%e3.23 ###
%
%e3.13 #&#
\begin{equation}
\label{eduhamel2} u_2 = G\ast(f\delta_{t=0}) + G\ast(-
\psi_{c_2}v_2),
\end{equation}
where we recall that $\ast$ denotes convolution on $\mathbb{R}_+
\times\mathbb{R}$. Subtracting the two equations gives
\[
u_1-u_2 = G\ast \bigl((\psi_{c_2}-
\psi_{c_1})v_1 + \psi_{c_2}(v_2-v_1)
\bigr).
\]
Bounding in terms of the $\sup$ norm and using the fact that
\[
\bigl|\psi \bigl(x-c_1(t) \bigr)-\psi \bigl(x-c_2(t) \bigr)\bigr|
\leq \|\psi_x\|_{L_x^\infty}\bigl|c_1(t)-c_2(t)\bigr|,
\]
we have
\begin{eqnarray*}
&& \bigl|u_1(x,t)-u_2(x,t)\bigr|
\\
&&\qquad \leq \int_0^t
\int_{\mathbb{R}}G(x-y,t-s)\bigl|\psi_{c_1}(y,s)-
\psi_{c_2}(y,s)\bigr|\bigl|v_1(y,s)\bigr|\,dy\,ds
\\
&&\quad\qquad{}+ \int_0^t\int_{\mathbb{R}}G(x-y,t-s)\bigl|
\psi_{c_2}(y,s)\bigr|\bigl|v_2(y,s)-v_1(y,s)\bigr|\,dy\,ds
\\
&&\qquad \leq\|\psi_x\|_{L_x^\infty}\|v_1\|_{L^\infty L_t^\infty([0,T])}
\| c_1-c_2\|_{L_x^\infty} t
+ \|\psi\|_{L_x^\infty} \|v_1-v_2
\|_{L_x^\infty L_t^\infty
([0,t])} t
\\
&&\qquad = BM\|c_1-c_2\|_{L_x^\infty} t +
\|v_1-v_2\|_{L_x^\infty L_t^\infty
([0,t])} t.
\end{eqnarray*}
Thus,
\[
\|u_1-u_2\|_{L_x^\infty L_t^\infty([0,T])} \leq B\|c_1-c_2
\| _{L_x^\infty} T + \|v_1-v_2\|_{L_x^\infty L_t^\infty([0,T])} T,
\]
so for small enough $T$ we see that (\ref{econtractionu}) holds.
\end{pf}

%
%
%le3.12 #&#
\begin{lemma}\label{lcontriactionux}
Suppose that $(v_1,b_1), (v_2,b_2)\in\Gamma^T_{\mathit{MNPALb}_0}$. Set
$(u_1,c_1)=\Phi((v_1,b_1))$ and $(u_2,c_2)=\Phi((v_2,b_2))$. For any
$\varepsilon>0$ there exists $T>0$ such that
%
%
%
%e3.14 #&#
\begin{equation}
\label{econtractionux} \|\partial_xu_1-
\partial_xu_2\|_{L_x^\infty L_t^\infty([0,T])}\leq\varepsilon
\bigl\|(v_2, b_2) - (v_1, b_1)\bigr\|_T.
\end{equation}
\end{lemma}

\begin{pf}
Differentiating (\ref{eubv}) with respect to $x$,
%
%e3.15 #&#
\begin{equation}\label{euxbvx}
\cases{\displaystyle \biggl(\partial_t-
\frac{\partial_{xx}}{2} \biggr)\,\partial_xu_1(x,t) =-
\psi_{x,c_1}(x,t)v_1(x,t)-\psi_{c_1}(x,t)\,\partial_xv_1(x,t),
\vspace*{3pt}\cr \displaystyle
\hspace*{112pt} x\in\mathbb{R}, t>0,
\vspace*{3pt}\cr \displaystyle
\biggl(\partial_t-\frac{\partial_{xx}}{2} \biggr)\,\partial_xu_2(x,t)
=-\psi_{x,c_2}(x,t)v_2(x,t)- \psi_{c_2}(x,t)\,\partial_xv_2(x,t),
\vspace*{3pt}\cr \displaystyle
\hspace*{112pt}x\in \mathbb{R}, t>0,
\vspace*{3pt} \cr
\partial_xu_1(x,0) = f_x(x),\qquad x \in\mathbb{R},
\vspace*{3pt}\cr
\partial_xu_2(x,0) = f_x(x),\qquad x \in \mathbb{R}.}\hspace*{-35pt}
\end{equation}
Via Duhamel's formula,
%
%e3.16 #&#
\begin{eqnarray}\label{eduhamelx1}
\partial_x u_1 &=& G
\ast(f_x \delta_{t=0})
\nonumber\\[-8pt]\\[-8pt]
&&{} + G\ast \bigl(-\psi_x
\bigl(\cdot -c_1(\cdot) \bigr)v_1-\psi \bigl(
\cdot-c_2( \cdot) \bigr)\,\partial_xv_1 \bigr)\nonumber
\end{eqnarray}
and
%
%
%e3.27 ###
%
%e3.17 #&#
\begin{eqnarray}\label{eduhamelx2}
\partial_x u_1 &=& G
\ast(f_x \delta_{t=0})
\nonumber\\[-8pt]\\[-8pt]
&&{} + G\ast \bigl(-\psi_x
\bigl(\cdot -c_2(\cdot) \bigr)v_2-\psi \bigl(
\cdot-c_2( \cdot) \bigr)\,\partial_xv_2 \bigr).\nonumber
\end{eqnarray}
Subtracting and rearranging,
\begin{eqnarray*}
&& (\partial_x u_1 - \partial_x u_2) (x,t)
\\
&&\qquad = \int_0^t\int_{\mathbb
{R}}G(x-y,t-s) \bigl[\psi_{x,c_2}v_2(y,s)-
\psi_{x,c_1}v_1(y,s) \bigr]\,dy\,ds
\\
&&\qquad\quad{}+\int_0^t\int_{\mathbb{R}}G(x-y,t-s)
\bigl[\psi_{c_2}\,\partial_xv_2(y,s)-
\psi_{c_1}\,\partial_xv_1(y,s) \bigr]\,dy\,ds
\\
&&\qquad = \int_0^t\int_{\mathbb{R}}G(x-y,t-s)
\bigl[\psi_{x,c_2}v_2(y,s) - \psi_{x,c_2}v_1(y,s)
\bigr]\,dy\,ds
\\
&&\quad\qquad{}+ \int_0^t\int_{\mathbb{R}}G(x-y,t-s)
\bigl[\psi_{x,c_2}v_1(y,s) - \psi_{x,c_1}v_1(y,s)
\bigr]\,dy\,ds
\\
&&\quad\qquad{}+\int_0^t\int_{\mathbb{R}}G(x-y,t-s)
\bigl[\psi_{c_2}\,\partial_xv_2(y,s)-
\psi_{c_2}\,\partial_xv_1(y,s) \bigr]\,dy\,ds
\\
&&\quad\qquad{}+\int_0^t\int_{\mathbb{R}}G(x-y,t-s)
\bigl[\psi_{c_2}\,\partial_xv_1(y,s) -
\psi_{c_1}\,\partial_xv_1(y,s) \bigr]\,dy\,ds.
\end{eqnarray*}
Using estimates similar to those in the proof of Lemma
\ref{lcontriactionu},
\begin{eqnarray*}
&& \|\partial_xu_1-\partial_xu_2
\|_{L_x^\infty L_t^\infty([0,T])}
\\
&&\qquad \leq BM\|v_2-v_1\|_{L_x^\infty
L_t^\infty([0,T])}T
+ CM\|c_2-c_1\|_{L_t^\infty([0,T])} T
\\
&&\quad\qquad{}+\|\partial_xv_2-\partial_xv_1
\|_{L_x^\infty L_t^\infty([0,T])} T
+BN\|c_2-c_1\|_{L_t^\infty([0,T])}T
\\
&&\qquad = BMT\|v_2-v_1\|_{L_x^\infty L_t^\infty([0,T])}
+(CM+BN)T\|c_2-c_1\|_{L_t^\infty([0,T])}
\\
&&\quad\qquad{}+ T\|\partial_xv_2-\partial_xv_1
\|_{L_x^\infty L_t^\infty([0,T])},
\end{eqnarray*}
so for $T$ small we recover (\ref{econtractionux}).
\end{pf}

%
%
%le3.13 #&#
\begin{lemma}\label{lcontriactionuxx}
Suppose that $(v_1,b_1)$, $(v_2,b_2)\in\Gamma^T_{\mathit{MNPALb}_0}$. Set
$(u_1,c_1)=\Phi((v_1,b_1))$ and $(u_2,c_2)=\Phi((v_2,b_2))$. For any
$\varepsilon>0$ there exists $T>0$ such that
%
%
%e3.28 ###
%
%e3.18 #&#
\begin{equation}
\label{econtractionuxx}
\|\partial_{xx}u_1-\partial_{xx}u_2\|_{L_x^\infty L_t^\infty
([0,T])}\leq \varepsilon
\bigl\|(v_2, b_2) - (v_1, b_1)\bigr\|_T.
\end{equation}
\end{lemma}

\begin{pf}
Differentiating (\ref{eubv}) twice with respect to $x$,
%
%
%e3.29 ###
%
%e3.19 #&#
\begin{equation}\label{euxxbvxx}
\cases{\displaystyle \biggl(\partial_t-
\frac{\partial_{xx}}{2} \biggr)\,\partial_{xx}u_1
=- \psi_{xx,c_1}v_1-2\psi_{x,c_1}\,\partial_xv_1-
\psi_{c_1}\,\partial_{xx}v_1,
\vspace*{2pt}\cr
\hspace*{121pt} x\in \mathbb{R}, t>0,
 \vspace*{2pt}\cr
\displaystyle \biggl(\partial_t-
\frac{\partial_{xx}}{2} \biggr)\,\partial_{xx}u_2
=- \psi_{xx,c_2}v_2-2\psi_{x,c_2}\,\partial_xv_2-
\psi_{c_2}\,\partial_{xx}v_2,
\vspace*{2pt}\cr
\hspace*{121pt} x\in \mathbb{R}, t>0,
\vspace*{2pt} \cr
\partial_{xx}u_1(x,0)
= f_{xx}(x),\qquad x \in\mathbb{R},
\vspace*{2pt}\cr
\partial_{xx}u_2(x,0) = f_{xx}(x),\qquad x \in \mathbb{R}.}
\end{equation}
Duhamel's formula and similar manipulations to Lemmas
\ref{lcontriactionu} and \ref{lcontriactionux} give
\begin{eqnarray*}
&& \|\partial_{xx}u_1-\partial_{xx}u_2
\|_{L_x^\infty L_t^\infty([0,T])}
\\
&&\qquad \leq CM\|v_2-v_1\|_{L_t^\infty([0,T]) L_x^\infty}T
\\
&&\quad\qquad{}+FM\|c_2-c_1\|_{L_t^\infty([0,T])} T
\\
&&\quad\qquad{}+2B\|\partial_xv_2-\partial_xv_1
\|_{L_x^\infty L_t^\infty([0,T])
} T
\\
&&\quad\qquad{}+2CN\|c_2-c_1\|_{L_t^\infty([0,T])}T
\\
&&\quad\qquad{}+\|\partial_{xx}v_2-\partial_{xx}v_1
\|_{L_x^\infty L_t^\infty
([0,T]) } T
\\
&&\quad\qquad{}+BP\|c_2-c_1\|_{L_t^\infty([0,T])}T
\\
&&\qquad = CMT\|v_2-v_1\|_{L_x^\infty L_t^\infty([0,T]) }
\\
&&\quad\qquad{}+ 2BT\|\partial_xv_2-\partial_xv_1
\|_{L_x^\infty L_t^\infty
([0,T])}
\\
&&\quad\qquad{}+ T\|\partial_{xx}v_2-\partial_{xx}v_1
\|_{L_x^\infty L_t^\infty
([0,T])}
\\
&&\quad\qquad{}+(FM+2CN+BP)T\|c_2-c_1\|_{L_t^\infty([0,T])},
\end{eqnarray*}
so when $T>0$ is small, (\ref{econtractionuxx}) holds.
\end{pf}

%
%
%th3.14 #&#
\begin{theorem}[(Local existence and uniqueness)]
\label{texistence} Suppose that the conditions of
Theorem~\ref{tglobalEandU} hold. Then, there exists a time $T>0$ such
that the system
\[
\cases{\displaystyle u_t(x,t) = \frac{1}{2}u_{xx}(x,t)
- \psi \bigl(x-b(t) \bigr)u(x,t), &\quad$x \in\mathbb{R}$, $t > 0$, \vspace*{3pt}
\cr
u(x,0) =f(x), &\quad$x \in\mathbb{R}$, \vspace*{3pt}
\cr
\displaystyle
b'(t) = \frac{g(t) + g'(t) - \langle\phi_b,u\rangle- 1/2 \langle
\psi_{x,b},u_x\rangle}{\langle\psi_{x,b},u\rangle}, &\quad$t > 0$, \vspace*{3pt}
\cr
b(0)=b_0,}
\]
has a unique solution $(u,b)\in C_x^2(\mathbb{R})C^1_t([0,T])\times
C^1([0,T])$.
\end{theorem}

\begin{pf}
Note there exist strictly positive constants $A, M, N$ and $P$ such
that $b_0\in[-\frac{A}{4},\frac{A}{4} ]$, $f(x)\geq L>0$, when
$x\in[-A,A]$, $\|f\|_{L^\infty(\mathbb{R})}\leq M$,\break  $\|f_x\|
_{L^\infty(\mathbb{R})}\leq N/2$, and
$\|f_x\|_{L^\infty(\mathbb{R})}\leq P/2$. Putting all the estimates
from the above lemmas together we have that, if $0<\varepsilon<1$ is
fixed, then for $T>0$ small enough,
\[
\bigl\|(u_2,c_2)-(u_1,c_1)\bigr\|\leq \varepsilon\bigl\|(v_2,b_2)-(v_1,b_1)
\bigr\|.
\]
Thus there exists a $T>0$ such that the map $\Phi\dvtx\Gamma
^T_{\mathit{MNPALb}_0}\rightarrow\Gamma^T_{\mathit{MNPALb}_0}$ is a contraction. Since
$\Gamma^T_{\mathit{MNPALb}_0}$ is a closed subset of the Banach space
$\mathcal{L}^T$, the contraction mapping theorem gives that there
exists a unique fixed point, that is, a pair $(u,b)\in
C_x^2(\mathbb{R})C_t([0,T])\times C([0,T])$ with $b(0)=b_0$ such that
%
%
%e3.30 ###
%
%e3.20 #&#
\begin{equation}\label{e_ub'}
\cases{\displaystyle u_t(x,t) = \frac{1}{2}u_{xx}(x,t)
- \psi \bigl(x-b(t) \bigr)u(x,t), \vspace*{2pt}
\cr
u(x,0) =f(x), \vspace*{2pt}
\cr
\displaystyle b'(t) = \frac{g(t) + g'(t) - \langle\phi_{b},u\rangle
- 1/2 \langle
\psi_{x,b},u_x\rangle}{\langle\psi_{x,b},u\rangle}, \vspace*{2pt}
\cr
b(0)=b_0.}
\end{equation}

We can now argue that our fixed point $(u,b)$ has more smoothness than
it seems a priori. The third equation in (\ref{e_ub'}) implies that $b$
must be continuously differentiable with a bounded derivative. This,
together with the first equation from (\ref{e_ub'}) then tells us that
$u$ has a continuous derivative in time. Therefore, we must have
$(u,b)\in C_x^2(\mathbb{R})C^1_t([0,T])\times C^1([0,T])$.
\end{pf}

%
%
%co3.15 #&#
\begin{corollary} Assume the hypotheses of
Theorem \ref{texistence} and the extra conditions
%
%
%e3.31 ###
%
%e3.21 #&#
\begin{equation}
\label{econditions} \cases{\displaystyle G(0) = \int_{\mathbb{R}}
f(x)\,dx, \vspace*{2pt}
\cr
\displaystyle-g(0) =\int_{\mathbb{R}}\psi
\bigl(x-b(0) \bigr)f(x)\,dx, \vspace*{2pt}
\cr
0<-g(t) < G(t), &\quad$t
\in[0,T]$.}
\end{equation}
Then, there exists a time $T>0$ such that the system
\[
\cases{\displaystyle u_t(x,t) = \frac{1}{2}u_{xx}(x,t)
- \psi \bigl(x-b(t) \bigr)u(x,t), &\quad$x \in\mathbb{R}$, $0<t<T$, \vspace*{2pt}
\cr
u(x,0) =f(x), &\quad$x \in\mathbb{R}$, \vspace*{2pt}
\cr
\displaystyle G(t) =
\int_{\mathbb{R}} u(x,t)\,dx, &\quad$t\in[0,T]$,}
\]
has a unique solution $(u,b)\dvtx\mathbb{R}\times[0,T] \to\mathbb{R}$.
Furthermore, $u\in C_x^2(\mathbb{R})C^1_t([0,T])$ and $b\in
C^1([0,T])$.
\end{corollary}

\begin{pf}
First note that by Lemma \ref{lIVP} we have that $b(0)$ is uniquely
determined. From Theorem \ref{texistence} we have that there exist
unique $u$, $b$\vadjust{\goodbreak} with $u\in C_x^2(\mathbb{R})C^1_t([0,T])$ and $b\in
C^1([0,T])$ satisfying the PDE and having everywhere in $[0,T]$
\[
b'(t) = \frac{g(t) + g'(t) - \langle\phi_b,u\rangle- 1/2 \langle
\psi_{x,b},u_x\rangle}{\langle\psi_{x,b},u\rangle}.
\\
\]

Set $F(t):=G(t)-\int_{\mathbb{R}} u(x,t)\,dx$ and note that the first
two conditions from~(\ref{econditions}) yield, together with the
PDE, $F_t(0)=F(0)=0$. The function $F$ belongs to $C^1([0,T])$, and
$F_t$ belongs to $C([0,T])$. The above equation for $b'$ is equivalent,
after using the PDE, to
\[
F_{tt}(t) - F_t(t)=0, \qquad t\in[0,T].
\]
Integrating and using the fundamental theorem of calculus, we get
\[
F_t(t) - F(t) = F_t(0)-F(0)=0, \qquad t\in[0,T].
\]
The unique solution to this differential equation is $F(t)=Ce^t$ for
some constant $C\in\mathbb{R}$. This together with $F(0)=0$ yields
$F(t)=0$ for $t\in[0,T]$. Thus
\[
G(t) = \int_{\mathbb{R}} u(x,t)\,dx, \qquad t\in[0,T].
\]
Then, taking a derivative and using the PDE,
\[
-g(t)=\int_{\mathbb{R}}\psi \bigl(x-b(t) \bigr)u(x,t)\,dx, \qquad t
\in[0,T].
\]
Because $|\psi(x)|\leq1$ for $x\in\mathbb{R}$, $\psi=0$ for $x\geq
h$ and $u(x,t)> 0$, we see that
\[
0<\int_{\mathbb{R}}\psi \bigl(x-b(t) \bigr)u(x,t)\,dx =-g(t)< \int
_{\mathbb
{R}}u(x,t)\,dx = G(t).
\]\upqed
\end{pf}

%s4 ###
%s4 #&#
\section{Discontinuous killing}

Next, we consider
the existence of a barrier when killing is done nonsmoothly.
That is, we ask whether there exists a function $b$ such that, for a
given $G$,
%
%
%e4.1 ###
%
%e4.1 #&#
\begin{equation}
\label{ehardbarrier} G(t) = \int_{\mathbb{R}} \mathbb{E} \biggl[\exp
\biggl(-\int_0^t \mathbf{1}_{(-\infty,0]}
\bigl(x + B_u - b(u) \bigr)\,du \biggr) f(x) \biggr]\,dx.
\end{equation}
Note that $\int_0^t \mathbf{1}_{(-\infty,0]}(x + B_u - b(u))\,du$ is
the time during the interval $[0,t]$ spent by a Brownian motion started
at $x$ below the barrier $b$.

%
%
%th4.1 #&#
\begin{theorem}
\label{Thardbarrier} There exists a function $b$ such that, for a
given, twice continuously differentiable $G$ satisfying
$0<-g(t)/G(t)<1$, $t\geq0$, equation (\ref{ehardbarrier}) holds for
all $t\geq0$.
\end{theorem}

\begin{pf}
Let $\phi$ be a smooth decreasing function supported on $[0,1]$ with
$\int_{\mathbb{R}} \phi(x)\,dx = 1$. Put
\[
\underline{\psi}_\varepsilon(x) = \int_x^\infty
\phi \bigl((y - \varepsilon)/\varepsilon \bigr) (1/\varepsilon)\,dy
\]
and
\[
\overline{\psi}_\varepsilon(x) = \int_x^\infty
\phi(y/\varepsilon) (1/\varepsilon)\,dy,
\]
so that
%
%
%e4.2 ###
%
%e4.2 #&#
\begin{equation}
\label{e*1} \underline{\psi}_\varepsilon(x) \leq 1\{x \leq 0\} \leq
\overline{ \psi}_\varepsilon.
\end{equation}
Note also that
%
%e4.3 #&#
\begin{equation}
\label{e**1} \underline{\psi}_\varepsilon(x)\mbox{ increases with }
\varepsilon\qquad\mbox{for all } x
\end{equation}
and
%
%e4.4 #&#
\begin{equation}
\label{e**2} \overline{\psi}_\varepsilon(x)\mbox{ decreases with }
\varepsilon\qquad\mbox{for all } x.
\end{equation}
Let $\underline{b}_\varepsilon$ and $\overline{b}_\varepsilon$ be the
two barriers corresponding to $\underline{\psi}_\varepsilon(x)$ and
$\overline{\psi}_\varepsilon$. The existence and uniqueness of these
two barriers follows by Theorem \ref{tglobalEandU}. From (\ref{e*1}) we
have that
\[
\overline{b}_\varepsilon(t) \leq \underline{b}_\varepsilon(t)
\]
for all $t$ and from (\ref{e**1}), (\ref{e**2}) that
\[
\overline{b}_\varepsilon(t) \mbox{ is increasing in } \varepsilon\qquad\mbox{for each } t
\]
and
\[
\underline{b}_\varepsilon(t) \mbox{ is decreasing in } \varepsilon\qquad\mbox{for each } t.
\]
Put
\[
\overline{b}_*(t) = \lim_{\varepsilon\downarrow0} \overline {b}_\varepsilon(t)
\]
and
\[
\underline{b}_*(t) = \lim_{\varepsilon\downarrow0} \underline
{b}_\varepsilon(t).
\]
Then
%
%e4.5 #&#
\begin{equation}
\label{e***} \overline{b}_*(t) \leq \underline{b}_*(t)
\end{equation}
and both of these functions give a stopping time with the correct
distribution for the case where $\psi$ is the indicator of
$(-\infty,0]$. Because of (\ref{e***}), it must be the case that
$\overline{b}_*(t) = \underline{b}_*(t)$ for Lebesgue almost all $t$.
\end{pf}

%s5 ###
%s5 #&#
\section{Pricing claims}\label{Spricing}

Suppose that the asset price $(X_t)_{t \geq 0}$ is a geometric Brownian
motion given by
%
%e5.1 #&#
\begin{equation}
\label{estock} \frac{dX_t}{X_t} = \mu\,dt + \sigma\,dW_t,\vadjust{\goodbreak}
\end{equation}
where
%the drift $\mu$ incorporates the risk free interest rate, $\sigma$ is
%the volatility and
$(W_t)_{t \geq 0}$ is a standard Brownian motion. We model default
using a diffusion $(Y_t)_{t \geq 0}$, where
%
%
%e5.2 ###
%
%e5.2 #&#
\begin{equation}
\label{eBMdefault} dY_t = dB_t,
\end{equation}
with $(B_t)_{t \geq 0}$ another standard Brownian motion. We assume
that the Brownian motions $W$ and $B$ are correlated with correlation
$-1\leq\rho\leq1$; that is, the cross-variation of the two processes
satisfies
\[
[B,W]_t=\rho t, \qquad t \geq 0.
\]
We can assume without loss of generality that for two independent
Brownian motions $B', B''$ we have
\[
\cases{\displaystyle W_t = B_t',
\vspace*{2pt}
\cr
B_t = \rho B_t'+\sqrt{1-
\rho^2}B_t''.}
\]
In the following we will look at pricing contingent claims with a fixed
maturity $T>0$ and payoff of the form
\[
F(X_T)1\{\tau>T\}
\]
for the random time
\[
\tau:=\inf \biggl\{t>0\dvtx\lambda\int_0^t
\psi \bigl(Y_s-b(s) \bigr)\,ds>U \biggr\},
\]
where $U$ is an independent exponentially distributed random variable
with mean one.

Note that
\[
\mathbb{E}^x \bigl[F(X_T)1\{\tau>T\} \bigr] =
\mathbb{E}^x \biggl[F(X_T)\exp \biggl(-\lambda\int
_0^T\psi \bigl(Y_s-b(s) \bigr)\,ds
\biggr) \biggr].
\]
More generally, we will be interested in expressions of the form
\begin{eqnarray*}
&& \mathbb{E}^x \bigl[F(X_T) 1\{\tau>T\} \mid (X_s)_{0\leq s\leq t}, \tau>t \bigr]
\\
&&\qquad= \mathbb{E}^x \biggl[F(X_T)\exp \biggl(-\lambda
\int_t^T\psi \bigl(Y_s-b(s) \bigr)
\,ds \biggr) \Bigm| (X_s)_{0\leq s\leq t}, \tau>t \biggr],
\end{eqnarray*}
which we interpret as the price of the payoff at time $0\leq t\leq T$
given that default has not yet occurred.

Consider the Markov process $Z =(X,Y,V)$, where $X$, $Y$ are as above,
and $V$ is a process that, when started at $v$ is at $v+t$ after $t$
units of time, that is, $V_t = V_0+t$. The generator of $Z$ is
\[
A = (1/2) \sigma^2 x^2 D_x^2 +
\mu x D_x + (1/2) D_y^2 + \rho\sigma x
D_x D_y+ D_v.
\]
We want to compute
\[
\mathbb{E}^{(x,y)} \bigl[F(X_T) e^{-\int_0^T \lambda\psi(Y_s -
b(s))\,ds} \bigr]=
\mathbb{E}^{(x,y,0)} \bigl[F(X_T) e^{-\int_0^T
\lambda\psi(Y_s - b(V_s))\,ds} \bigr].
\]
The Feynman--Kac formula says that the solution to the PDE
%
%e5.3 #&#
\begin{equation}
\label{ediffusionKac} \cases{\displaystyle D_t u(x,y,v,t) = A
u(x,y,v,t) - \lambda\psi \bigl(y - b(v) \bigr) u(x,y,v,t), \vspace*{2pt}
\cr
u(x,y,v,0) = F(x),}
\end{equation}
satisfies
\[
\mathbb{E}^{(x,y)} \biggl[F(X_T) \exp \biggl(-\int
_0^T \lambda\psi \bigl(Y_s - b(s)
\bigr)\,ds \biggr) \biggr] = u(x,y,0,T).
\]
Thus, if we assume the Brownian motion $Y$ has a random starting point
$Y_0$ with density $f$, that is, independent of $(Y_t - Y_0)_{t \geq
0}$, then
\[
\mathbb{E}^{x} \biggl[F(X_T) \exp \biggl(-\int
_0^T \lambda\psi \bigl(Y_s - b(s)
\bigr)\,ds \biggr) \biggr] = \int_{\mathbb{R}}u(x,y,0,T)f(y)\,dy.
\]
Using this and the Markov property, one can find the function
$K(x,y,t)$ satisfying
\begin{eqnarray*}
&& K(X_t,Y_t,t)
\\
&&\qquad =\mathbb{E}^x \biggl[F(X_T)\exp
\biggl(-\lambda\int_t^T \psi \bigl(Y_s-b(s) \bigr)\,ds \biggr) \Bigm|
(X_s)_{0\leq s\leq t}, (Y_s)_{0\leq s\leq t}, \tau>t \biggr].
\end{eqnarray*}
The price at time $t$, given that we know the history of the price
process $X_t$ and that default has not happened up to time $t$, is
\begin{eqnarray*}
&& \mathbb{E} \bigl[F(X_T)1\{\tau>T\} \mid (X_s)_{0\leq s\leq t},
\tau>t \bigr]
\\
&&\qquad = \mathbb{E} \bigl[K(X_t,Y_t,t) \mid (X_s)_{0\leq s\leq t}, \tau>t \bigr]
\\
&&\qquad = \frac{\mathbb{E}[K(X_t,Y_t,t)1\{\tau>t\} \mid (X_s)_{0\leq s\leq
t}]}{\mathbb{E}[1\{\tau>t\} \mid (X_s)_{0\leq s\leq t}]}.
\end{eqnarray*}

It follows from the SDE for $X$ that
\[
B_t' = W_t = \frac{1}{\sigma} \biggl[
\log X_t - \log X_0 + \biggl(\frac{\sigma^2}{2} - \mu
\biggr) t \biggr],
\]
so if we observe the asset price $X$, then we can reconstruct the
standard Brownian motion $B'$. On the other hand,
\[
X_t = X_0 \exp \biggl(\sigma B_t'
- \biggl(\frac{\sigma^2}{2} - \mu \biggr) t \biggr).
\]

Now,
\begin{eqnarray*}
&& \mathbb{E} \bigl[K(X_t,Y_t,t)1\{\tau>t \}
\mid(X_s)_{0\leq s\leq t} \bigr]
\\
&&\qquad= \mathbb{E} \biggl[K \biggl(X_0 \exp \biggl( \sigma
B_t' - \biggl(\frac{\sigma^2}{2} - \mu \biggr) t
\biggr), Y_0 + \rho B_t'+\sqrt{1-\rho^2}B_t'', t \biggr)
\\
&&\hspace*{42pt}{} \times\mathbf{1} \biggl\{ \int_0^t
\psi \bigl(Y_0 + \rho B_s'+\sqrt{1-\rho^2}B_s'' - b(s) \bigr)\,ds
\leq U \biggr\}
\Bigm| X_0, \bigl(B_s' \bigr)_{0\leq s\leq t} \biggr].
\end{eqnarray*}
We therefore want to be able to compute for a function $c\dvtx\mathbb
{R}_+ \to\mathbb{R}$ the conditional expected value
\begin{eqnarray*}
&& \mathbb{E} \biggl[K \biggl(X_0 \exp \biggl( \sigma c(t) - \biggl(
\frac{\sigma^2}{2} - \mu \biggr) t \biggr), Y_0 + \rho c(t) +
\sqrt{1- \rho^2}B_t'', t
\biggr)
\\
&&\hspace*{32pt}{}\times\mathbf{1} \biggl\{ \int_0^t
\psi \bigl(Y_0 + \rho c(s) +\sqrt{1-\rho^2}B_s''
- b(s) \bigr)\,ds \leq U \biggr\}\Bigm| X_0 \biggr]
\\
&&\qquad= \mathbb{E} \biggl[K \biggl(X_0 \exp \biggl( \sigma c(t) -
\biggl(\frac{\sigma^2}{2} - \mu \biggr) t \biggr), Y_0 + \rho c(t) +
\sqrt{1-\rho^2}B_t'', t\biggr)
\\
&&\hspace*{66pt}{}\times\exp \biggl( - \int_0^t \psi
\bigl(Y_0 + \rho c(s) +\sqrt{1-\rho^2}B_s''
- b(s) \bigr)\,ds \biggr)\Bigm| X_0 \biggr]
\end{eqnarray*}
with $(B_t'')_{t \geq 0}$ a standard Brownian motion independent of
$X_0$. We can do this using Feynman--Kac.

The denominator in the formula for the price at time $t$ is a special case
of the numerator we have just calculated
with $K \equiv1$, and it can be dealt with in the same way.

%
%
%f1 ###
%f1 #&#
\begin{figure}

\includegraphics{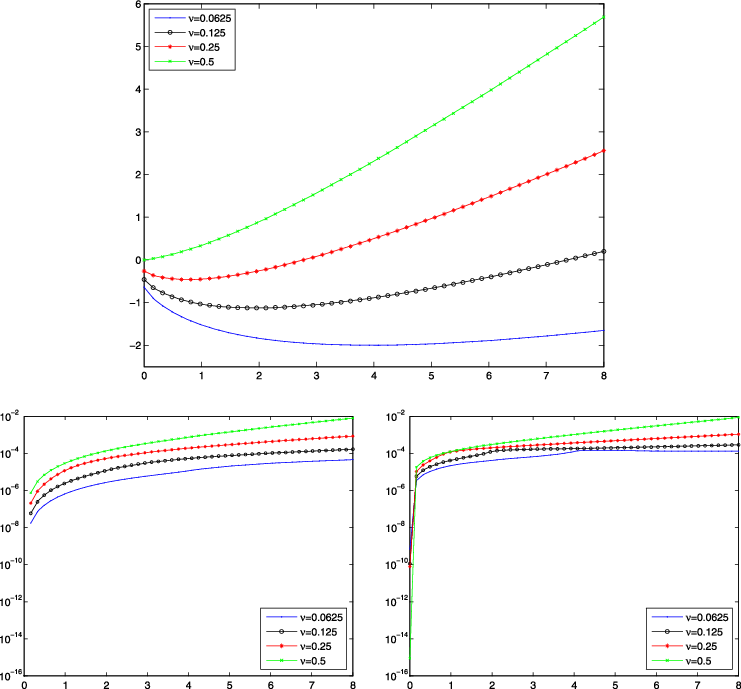}

\caption{This figure displays the results of the numerical experiments
described in Section~\protect\ref{Snumerical}. We fix the standard
deviation for the initial distribution of the credit index process $Y$
to be $\sigma=0.25$ and the killing parameter to be $\lambda=1$.
The first row gives the barriers for the rate parameters
$\nu=0.0625, 0.125, 0.25, 0.5$ of the exponential default time distribution.
The first (resp., second) panels in the second row give the relative
errors between
the actual survival function values $G(t)$ [resp., the actual hazard
function values $-g(t)/G(t)$] and the numerically
computed ones; see the text for details.}\label{figur}
\end{figure}

We have thus observed that computing the price of a contingent claim
reduces to solving certain PDEs with coefficients depending on the path
of the asset price.

%s6 ###
%s6 #&#
\section{Numerical results}\label{Snumerical}
In this section we present the results of some\break  numerical experiments.
We solved the PDE/ODE system (\ref{ePDEODE}) using\break  the pseudo-spectral
implicit-explicit fourth order Runge--Kutta scheme\break  ARK4(3)6L[2]SA-ERK,
taking 8192 nodes and period 16, developed in \cite{MR03}. For the
function $\psi$ we used the Fej\'{e}r kernel of order 512 applied to
the indicator of the set $\{x \in\mathbb{R} \dvtx x < 0\}$; in other
words $\psi$ is the Ces\`{a}ro sum of the truncated Fourier series of
order 512 of the indicator of the set $\{x \in\mathbb{R} \dvtx x <
0\}$. The time horizon was taken to be $T=8$, the initial distribution
of the credit index process $Y$ was taken to be normal [$Y_0 \sim
N(0,\sigma^2)$ with standard deviation $\sigma=0.25$], and the time to
default was taken to have an exponential distribution [$G(t)=e^{-\nu
t}$ with rates $\nu=0.0625, 0.125, 0.25, 0.5$].

For the first experiment, we fix the killing parameter to $\lambda=1$.
We show the resulting barriers $b$ in Figure~\ref{figur}. We also show
the relative error between the survival function $G(t)$ and the
numerically computed value of $\int_{\mathbb{R}}u(x,t)\,dx$ [recall
(\ref{ePDE})], and the relative error between the hazard rate
$-g(t)/G(t)$ and the numerically computed value of
$\int_{\mathbb{R}}\psi(x-b(t))u(x,t)\,dx / \int_{\mathbb
{R}}u(x,t)\,dx$
[recall (\ref{egpsi})].

For the second experiment, we take the exponential rate to be
$\nu=0.125$ and the standard deviation to be $\sigma=0.25$. We look at
the graphs for when the killing parameter is $\lambda=1, 10, 50, 200$.
The barriers, together with the relative errors in the survival
functions and hazard rates are given in Figure~\ref{figur2}.

%
%
%f2 ###
%f2 #&#
\begin{figure}

\includegraphics{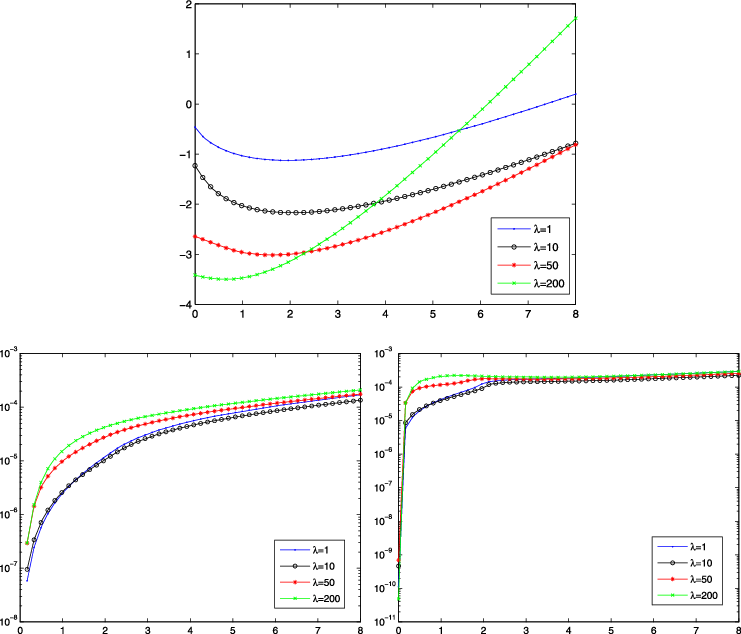}\vspace*{-3pt}

\caption{In this figure we fixed the standard deviation to $\sigma=0.25$
and the rate parameter to $\nu=0.125$. The first row gives the barriers
for the killing parameters $\lambda= 1, 10, 50, 200$. The first and second
panels in the second row give the relative errors for the survival function
(resp., the hazard function).}\label{figur2}\vspace*{-5pt}
\end{figure}

%s7 ###
%s7 #&#
\section{Calibrating the default distribution using CDS rates}\label{Scalibration}

For the sake of completeness, we review briefly the scheme proposed in
\cite{DP11} for determining the distribution of the time to default.

A credit default swap (CDS) is a contract between two parties.
The buyer of the swap makes a number of predetermined
payments until the moment of default.
The seller is liable to pay the unrecovered value of the underlying
bond in the event of a default before maturity.
Normalizing the notional value of the bond to $1$,
the seller's contingent payment is $1-R$,
where $R \in(0,1)$ is the recovery rate,
which we take to be constant.
The premium payments are made at a set of times $\{t_i\}$.
The maturities are a subset of the premium payment times;
that is, they are of the form $T_0=0$, $T_j=t_{k(j)}$, $j=1,\dots,n$.
For $j=1,\dots,n$ there is
an upfront premium $\pi_j^0$ and a running premium rate $\pi_j^1$
(having accrual factors $\delta_i$). Denote the
price at time zero of a zero coupon risk-free bond with maturity
$t_j$ by $p_0(t_j)$.
It follows from standard nonarbitrage arguments that
%
%
%e7.1 ###
%
%e7.1 #&#
\begin{eqnarray}\label{eCDS}
&& \pi_j^0 + \pi_j^1\sum_{i=k(j-1)}^{k(j)-1}\delta_ip_0(t_i)G(t_i)
\nonumber\\[-10pt]\\[-10pt]
&&\qquad = (1-R)\sum_{i=k(j-1)+1}^{k(j)} p_0(t_i)\bigl(G(t_{i-1})-G(t_i) \bigr),\nonumber
\end{eqnarray}
where $G(t)=\mathbb{P}\{\tau>t\}$ is the tail of the distribution of the
time to default.\vadjust{\goodbreak}

Suppose now that the
default distribution has piecewise constant hazard rate; that is, that
\[
G(t)=\exp \biggl(-\int_0^t h(s)\,ds \biggr),
\qquad t\geq0,
\]
where $h(s)=h_i$ for $s\in[T_i,T_{i+1})$. Given the market data
$(\pi_1^0,\pi_1^1), (\pi_2^0,\pi_2^1), \dots$ we can find, using
equation (\ref{eCDS}), the constants $h_0, h_1,\dots.$

We use the following procedure to find the barrier $b$. Set $\nu=h_0$
and $G(t)=e^{-\nu t}$. Given the initial density $f$, which we can
choose to be any strictly positive function $f$, that is, twice
continuously differentiable with bounded $f$, $f'$ and $f''$, we want
to find a barrier such that for $0\leq t\leq T=T_1$ we have
\[
e^{-\nu t} = \mathbb{E} \biggl[\int_{\mathbb{R}} f(x)\exp
\biggl(-\lambda\int_0^t \psi
\bigl(x+B_s-b(s) \bigr)\,ds \biggr)\,dx \biggr].
\]
This can be achieved by solving the ODE/PDE system (\ref{ePDEODE}).
Next, set $\nu_1=h_1$, $T=T_2-T_1$, $f_1(x)=\mathbb{E} [f(x)\exp
(-\lambda\int_0^{T_1} \psi(x+B_s-b(s))\,ds ) ]$, and find a barrier
with $b_1(0)=b(T_1)$ such that on $0\leq t< T=T_2-T_1$ we have
\[
e^{-\nu_1 t} = \mathbb{E} \biggl[\int_{\mathbb{R}}
f_1(x)\exp \biggl(-\lambda\int_0^t
\psi \bigl(x+B_s-b_1(s) \bigr)\,ds \biggr)\,dx \biggr].
\]
This procedure can be repeated until we find a function $b$ on
$[0,T_n]$, that is, continuously differentiable everywhere, except
perhaps the finite number of points~$T_1,\dots,T_n$.

%s8 ###
%s8 #&#
\section{Duhamel's formula}\label{ADuhamel}

For the sake of reference, we provide a statement of Duhamel's formula.
Given functions $v\dvtx\mathbb{R}\times\mathbb{R}_+ \to\mathbb
{R}$ and
$b\dvtx\mathbb{R}_+ \to\mathbb{R}$, the solution of
%
%
%e8.1 ###
%
%e8.1 #&#
\begin{equation}\label{e_ub}
\cases{\displaystyle \biggl(\partial_t-
\frac{\partial_{xx}}{2} \biggr)u =-\psi_b v, &\quad $x\in\mathbb{R}$,
$t>0$, \vspace*{2pt}
\cr
u(x,0) = f(x), &\quad$x \in\mathbb{R}$,}
\end{equation}
is given by
%
%e8.2 ###
%
%e8.2 #&#
\begin{eqnarray}
\label{eduhamel} u(x,t) &=& \bigl[G\ast(f\delta_{t=0}) \bigr](x,t) +
\bigl[G\ast(-\psi_{b}v) \bigr](x,t)\nonumber
\\
&=& \int_{\mathbb{R}}G(x-y,t)f(y)\,dy
\\
&&{} -\int_0^t \int_{\mathbb {R}}G(x-y,t-s)\psi_{b(s)}(y) v(y,s)\,dy\,ds,\nonumber
\end{eqnarray}
where
%
%
%e8.3 ###
%
%e8.3 #&#
\begin{equation}
\label{eheatkernel} G(x,t):=\frac{1}{\sqrt{2\pi t}}e^{-x^2/(2t)}, \qquad x \in
\mathbb{R}, t > 0.
\end{equation}

\section*{Acknowledgments} We are grateful to Professor Jon Wilkening
for helpful comments regarding Section~\ref{Snumerical} and to an
anonymous referee for suggestions which improved this paper.

%suskaldyti doi

% imsref loaded by linak, 2013-09-10 13:49:10
% imsref loaded by linak, 2013-09-10 13:57:11

\printaddresses

\end{document}